\documentclass[twoside]{article}

\usepackage[accepted]{aistats2019}

\usepackage[T1]{fontenc}
\usepackage[latin9]{inputenc}
\usepackage{color}
\usepackage[english]{babel}
\usepackage{verbatim}
\usepackage{float}
\usepackage{amsmath}
\usepackage{amsthm}
\usepackage{amssymb}
\usepackage{graphicx}
\usepackage{hyperref}

\makeatletter

\floatstyle{ruled}
\newfloat{algorithm}{tbp}{loa}
\providecommand{\algorithmname}{Algorithm}
\floatname{algorithm}{\protect\algorithmname}

\theoremstyle{plain}
\newtheorem{thm}{\protect\theoremname}
\theoremstyle{plain}
\newtheorem{assumption}[thm]{\protect\assumptionname}
\theoremstyle{plain}
\newtheorem{prop}[thm]{\protect\propositionname}
\theoremstyle{plain}

\usepackage{dsfont}
\usepackage{enumitem}
\setlist{nolistsep}

\@ifundefined{showcaptionsetup}{}{%
 \PassOptionsToPackage{caption=false}{subfig}}
\usepackage{subfig}
\makeatother

\providecommand{\assumptionname}{Assumption}
\providecommand{\lemmaname}{Lemma}
\providecommand{\propositionname}{Proposition}
\providecommand{\theoremname}{Theorem}

\begin{document}

\twocolumn[

\aistatstitle{Unbiased Smoothing using Particle Independent Metropolis--Hastings}
\aistatsauthor{ Lawrence Middleton  \And George Deligiannidis   \And Arnaud Doucet   \And Pierre E. Jacob  }
\aistatsaddress{ University of Oxford  \And  University of Oxford  \And University of Oxford \And   Harvard University  } ]

\begin{abstract}
We consider the approximation of expectations with respect to the
distribution of a latent Markov process given noisy measurements.
This is known as the smoothing problem
and is often approached with particle and Markov chain
Monte Carlo (MCMC) methods. These methods provide consistent but biased
estimators when run for a finite time.
We propose a simple way of coupling two MCMC chains
built using Particle Independent Metropolis--Hastings (PIMH)
to produce unbiased smoothing estimators.
Unbiased estimators are appealing in the context
of parallel computing, and facilitate the construction of confidence intervals.
The proposed scheme only requires access to off-the-shelf Particle
Filters (PF) and is thus easier to implement than recently
proposed unbiased smoothers. The approach
is demonstrated on a L\'evy-driven stochastic volatility model and a stochastic kinetic model.
\end{abstract}

\section{Introduction}

\subsection{State-space models, problem statement and review}
Let $t$ denote a discrete-time index. State-space models are defined
by a latent Markov process $(X_{t})_{t\ge1}$ and observation process
$(Y_{t})_{t\ge1}$, $(X_{t},Y_{t})$ taking values in a measurable
space $(\mathsf{X}\times\mathsf{Y},\mathcal{B}(\mathsf{X})\otimes\mathcal{B}(\mathsf{Y}))$
and satisfying
\[
X_{t+1}|\{X_{t}=x\}\sim f(\cdot|x),\enskip Y_{t}|\{X_{t}=x\}\sim g(\cdot|x),\enskip
\]
for $t\ge1$ with $X_{1}\sim\mu(\cdot)$. In the following we assume that $\mathsf{X}\subseteq\mathbb{R}^{d_{x}}$
and $\mathsf{Y}\subseteq\mathbb{R}^{d_{y}}$,
and use $f(\cdot|x)$, $g(\cdot|x)$ and $\mu(\cdot)$ to denote densities
with respect to the corresponding Lebesgue measure. State inference
given a realization of the observations $Y_{1:T}=y_{1:T}$ for some
fixed $T\in\mathbb{N}$ requires the posterior density
\begin{align*}
\pi\left(x_{1:T}\right):=p(x_{1:T}|y_{1:T})\qquad\qquad\qquad\qquad\\
\propto\mu(x_{1})g(y_{1}|x_{1})\prod_{t=2}^{T}f(x_{t}|x_{t-1})g(y_{t}|x_{t}),
\end{align*}
and expectations w.r.t. to this density. This is known as smoothing in the literature. For non-linear non-Gaussian state-space models,
this problem is complex as this posterior and its normalizing constant
$p(y_{1:T})$, often called somewhat abusively likelihood,
are intractable. We provide here a means to obtain unbiased
estimators of $\pi(h):=\int h(x_{1:T})\pi(x_{1:T})\mathrm{d}x_{1:T}$ for some function $h:\mathsf{X}^{T}\rightarrow\mathbb{R}$.

Particle methods return asymptotically consistent estimators which
are however biased for a finite number of particles. Similarly MCMC
kernels, such as the iterated Conditional Particle Filter (i-CPF)
and PIMH \cite{andrieu2010particle}, can be used for approximating
smoothing expectations consistently but are also biased for a finite
number of iterations. Additionally, although theoretical bounds on
the bias are available for particle \cite{delmoral2007bias}, PIMH \cite{andrieu2010particle} and
i-CPF \cite{chopin2015particle,andrieu2018uniform,lee2018coupled} estimators, these bounds are usually not sharp and/or rely on strong mixing assumptions which are not satisfied by most realistic models. Unbiased
estimators of $\pi(h)$ computed on parallel machines can be combined
into asymptotically valid confidence intervals
as either the number of machines, the time budget, or both go to infinity \cite{Glynn1991}.

Recently, it has been shown in \cite{jacob2017smoothing,lee2018coupled}
that it is possible to obtain such unbiased estimators by combining
the i-CPF algorithm with a debiasing scheme for MCMC algorithms proposed
initially in \cite{glynn2014exact} and further developed in \cite{jacob2017unbiased}.
These unbiased smoothing schemes couple two i-CPF kernels using common
random numbers and a coupled resampling scheme. After a brief review
of particle methods and of PIMH, we propose in Section \ref{sec:Methodology}
an alternative methodology relying on coupling two PIMH kernels.
The method is easily implementable as it does not require any modification of the
PF algorithm.
It can also
be used in scenarios where simulation from the Markov transition
kernel of the latent process involves a random number of random variables,
whereas the
methods proposed in \cite{jacob2017smoothing,lee2018coupled} would not be directly applicable to these settings. Additionally it
does not require being able to evaluate pointwise the transition density contrary to the coupled conditional backward sampling PF scheme of \cite{lee2018coupled}. Section \ref{sec:guidelines}
presents an analysis of the methodology when $T$ is large. In Section \ref{sec:Numerical-experiments}, the method
is demonstrated on a L\'evy-driven stochastic volatility
model and a stochastic kinetic model.\footnote{Code to reproduce figures is provided at \url{https://github.com/lolmid/coupled_pimh}}

\subsection{Particle methods}

Particle methods are often used to approximate smoothing expectations \cite{doucet2009tutorial,kantas2015particle}.
Such methods rely on sampling, weighting and resampling a set of $N$
weighted particles $(X_{t}^{i},W_{t}^{i})$, where $X_{t}^{i}\in\mathsf{X}$
denotes the value of the $i^{th}$ particle at iteration $t$ and
$W_{t}^{i}$ its corresponding normalized weight, i.e. $\sum_{i=1}^{N}W_{t}^{i}=1$.
Letting $q_{1}(x_{1})$, $q_{t}(x_{t}|x_{t-1})$, denote the proposal
density at time $t=1$ and at time $t\geq2$ respectively, weighting
occurs according to the following `incremental weights':
\begin{align*}
w_{1}(x_{1})&:=\frac{g(y_{1}|x_{1})\mu(x_{1})}{q_{1}(x_{1})}\quad\text{for}\quad t=1,\\
w_{t}(x_{t-1},x_{t})&:=\frac{g(y_{t}|x_{t})f(x_{t}|x_{t-1})}{q_{t}(x_{t}|x_{t-1})}\quad\text{for}\quad t\geq 2.
\end{align*}
We assume that $w_{1}(x_{1})>0$ and $w_{t}(x_{t-1},x_{t})>0$ for $t=2,...,T$ and all $x_{1:T}$. Pseudo-code for a standard PF is presented in Algorithm
\ref{alg:pf} where we let $r(\cdot|\mathbf{W}_{t})$, with $\mathbf{W}_{t}:=\left(W_{t}^{1},...,W_{t}^{N}\right)$, denote the resampling distribution, a probability distribution on
$[N]^{N}$ where $[N]:=\{1,...,N\}$. We say that a resampling scheme is unbiased if $\sum_{i=1}^{N}r(A_{t}^{i}=k|\mathbf{W}_{t})=NW_t^{k}$.
All standard resampling schemes -multinomial, residual and systematic- are unbiased \cite{doucet2009tutorial}.
This PF procedure outputs an approximation $p_{N}(y_{1:T})$ of $p(y_{1:T})$ and an approximation
$\pi_{N}(\mathrm{d}x_{1:T})$ of the smoothing distribution $\pi(\mathrm{d}x_{1:T})$.
Under weak assumptions, it can be shown that $p_{N}(y_{1:T})$,
resp. $\pi_{N}(h):=\sum_{i=1}^{N}W_{T}^{i}h\left(X_{1:T}^{i}\right)$,
is an asymptotically consistent (in $N$) estimator of $p(y_{1:T})$,
resp. of $\pi(h)$. However, whereas $p_{N}(y_{1:T})$ is unbiased
(\cite{Del_Moral_2004}, Section 7.4.1), $\pi_{N}(h)$ admits an asymptotic bias
of order $C/N$ for a constant $C$ which is typically impossible
to evaluate and for which only loose bounds are available under realistic
assumptions \cite{delmoral2007bias}. In the following by a call to
the PF, $(X_{1:T},p_{N})\sim\text{PF}$, we mean
a procedure which runs Algorithm \ref{alg:pf} and returns $p_{N}:=p_{N}(y_{1:T})$
(dependence on observations is notationally omitted) and a sample
from the approximate smoothing distribution $X_{1:T}\sim\pi_{N}\left(\cdot\right)$, i.e. output $X_{1:T}^i$ with probability $W_T^i$.

\begin{algorithm}
\caption{Particle Filter\label{alg:pf}}
\textsf{For $i\in[N]$, sample }$X_{1}^{i}\stackrel{i.i.d.}{\sim}q_{1}(\cdot)$\textsf{,
compute weights $W_{1}^{i}\propto w_{1}(X_{1}^{i})$ and set ${p}_{N}(y_{1})=\frac{1}{N}\sum_{i=1}^{N}w_{1}(X_{1}^{i})$.}
\textsf{For $2\le t\le T$:}
\begin{enumerate}
\item \textsf{Sample particle ancestors $A_{t-1}^{1:N}\sim r(\cdot|\mathbf{W}_{t-1})$.}
\item \textsf{For $i\in[N]$, sample $X_{t}^{i}\sim q_{t}(\cdot|X_{t-1}^{A_{t-1}^{i}})$
and set $X_{1:t}^{i}=\{X_{1:t-1}^{A_{t-1}^{i}}, X_{t}^{i}\}$.}
\item \textsf{Compute weights $W_{t}^{i}\propto w_{t}(X_{t-1}^{A_{t-1}^{i}},X_{t}^{i})$
and set $p_{N}(y_{1:t})=p_{N}(y_{1:t-1})\cdot\frac{1}{N}\sum_{i=1}^{N}w_{t}(X_{t-1}^{A_{t-1}^{i}},X_{t}^{i})$.}
\end{enumerate}
\textsf{Output $\pi_{N}(\mathrm{d}x_{1:T}):=\sum_{i=1}^{N}W_{T}^{i}\delta_{X_{1:T}^{i}}(\mathrm{d}x_{1:T})$
and $p_{N}(y_{1:T})$.}
\end{algorithm}

\subsection{PIMH method}

An alternative way to estimate $\pi(h)$ consists of using an MCMC scheme targeting $\pi$. The PIMH algorithm achieves this by building a Markov chain on an extended space admitting a stationary distribution $\bar{\pi}$ with marginal $\pi$ using the PF described in Algorithm \ref{alg:pf}
as proposal distribution \cite{andrieu2010particle}. Algorithm \ref{alg:pimh} provides pseudo-code
for sampling the PIMH kernel,  where $(X_{1:T}^{(n)},p_{N}^{(n)})$
denotes the current state of this Markov chain and $a\wedge b$ means $\min(a,b)$.

\begin{algorithm}
\caption{PIMH kernel $P((X_{1:T}^{(n)},p_{N}^{(n)}),(\cdot,\cdot))$\label{alg:pimh}}

\begin{enumerate}
\item \textsf{Sample $(X_{1:T}^{*},p_{N}^{*})\sim\text{PF}$. }
\item \textsf{Set $(X_{1:T}^{(n+1)},p_{N}^{(n+1)})=(X_{1:T}^{*},p_{N}^{*})$
with probability }$\alpha(p_{N}^{(n)},p_{N}^{*}):=1\wedge\frac{p_{N}^{*}}{p_{N}^{(n)}}.$
\item \textsf{Otherwise set $(X_{1:T}^{(n+1)},p_{N}^{(n+1)})=(X_{1:T}^{(n)},p_{N}^{(n)}).$}
\end{enumerate}
\end{algorithm}
Validity and convergence properties of PIMH rely on viewing
it as an Independent Metropolis--Hastings (IMH) sampler on an extended space. For ease of presentation, we detail this construction for an unbiased resampling scheme satisfying additionally $r(A_t^{i}=k|\mathbf{W}_{t})=W_{t}^{k}$ for all $i,k\in[N]$. The PF of Algorithm \ref{alg:pf}
implicitly defines a distribution $\psi$ over $N\times T$ particle
coordinates and $N\times(T-1)$ ancestors. We use $\zeta$ to denote
a sample from $\psi$, where $\zeta\in\mathcal{X}:=\mathsf{X}^{NT}\times\{1,...,N\}^{N(T-1)}$,
and the density of $\psi$ is given by
\[
\psi(\zeta)=\left(\prod_{i=1}^{N}q_{1}(x_{1}^{i})\right)\prod_{t=2}^{T}\left(r(a_{t-1}^{1:N}|\mathbf{w}_{t-1})\prod_{i=1}^{N}q(x_{t}^{i}|x_{t-1}^{a_{t-1}^{i}})\right).
\]
We let $b_{t}^{j}$ denote the index of the ancestor particle of $x_{1:T}^{j}$
at generation $t$, which may be obtained deterministically from the
ancestry, using $b_{t}^{j}=a_{t}^{b_{t+1}^{j}}$ with
$b_{T}^{j}=j$. From \cite{andrieu2010particle}, for $(j,\zeta)\in\{1,...,N\}\times\mathcal{X}$,
we express the target $\bar{\pi}(j,\zeta)$ of the resulting IMH sampler
as
\begin{align*}
\bar{\pi}(j,\zeta) & =\frac{\pi(x_{1:T}^{j})}{N^{T}}\frac{\psi(\zeta)}{q_{1}(x_{1}^{b_{1}^{j}})\prod_{t=2}^{T}r(b_{t-1}^{j}|\mathbf{w}_{t-1})q_{t}(x_{t}^{b_{t}^{j}}|x_{t-1}^{b_{t-1}^{j}})}
\end{align*}
Running a PF and sampling from $\pi_{N}$ corresponds to sampling from the proposal $\bar{q}(j,\zeta)=\psi(\zeta)W_{T}^{j}$
and \cite[Theorem 2]{andrieu2010particle} shows that $\bar{\pi}(j,\zeta)/\bar{q}(j,\zeta)=p_{N}(y_{1:T})/p(y_{1:T})$
where for a given $(j,\zeta)$ we understand $p_{N}(y_{1:T})$
as a deterministic map from $\{1,...,N\}\times\mathcal{X}$ to $\mathbb{R}^{+}$.
As a result, samples from $\bar{\pi}$ can be obtained asymptotically
through accepting proposals with probability $\alpha(p_{N}^{(n)},p_{N}^{*})$.
We can thus estimate $\pi(h)$ by averaging over iterations $h(X_{1:T}^{n})$.
As shown in \cite[Theorem 6]{andrieu2010particle},
we can also estimate $\pi(h)$ by averaging over iterations $\pi_{N}^{(n)}(h)=\sum_{i=1}^{N}W_{T}^{i,n}h (X_{1:T}^{i,n})$, hence reusing all the particle
system used to generate the \emph{accepted }proposal
at iteration $n$ as $\bar{\pi}(\pi_{N}(h))=\pi(h)$.
As such, although Algorithms \ref{alg:pimh} and \ref{alg:couplepimh}
are stated in terms of proposing and accepting
$(X_{1:T},p_{N})$, it is possible to consider them instead as proposing
and accepting $(J,\zeta)$. In \cite{leemurrayjohansen2019}, it is shown that $r(A_t^{i}=k|\mathbf{W}_{t})=W_{t}^{k}$ is unnecessary. We only need to use an unbiased resampling scheme to obtain a valid PIMH scheme. This is achieved by defining an alternative target $\bar{\pi}$ on $\{1,...,N\}\times\mathcal{X}$ such that $X_{1:T}^{J}\sim \pi$ under $\bar{\pi}$ and $\bar{\pi}(j,\zeta)/\bar{q}(j,\zeta)=p_{N}(y_{1:T})/p(y_{1:T})$ also hold. It is also established in \cite{leemurrayjohansen2019} that one can even used adaptive resampling procedures \cite{doucet2009tutorial,del2012adaptive}.

If we denote by $Z$ the error of the log-likelihood estimator, i.e.
$Z:=\log\{p_{N}(y_{1:T})/p(y_{1:T})\}$, the PIMH algorithm induces
a Markov chain with transition kernel $Q(z,\mathrm{d}z')$ given by
\begin{equation}
\left\{1\wedge\exp(z'-z)\right\} g(\mathrm{d}z')+\left\{ 1-\alpha(z)\right\} \delta_{z}(\mathrm{d}z'),\label{eq:PforZ}
\end{equation}
where $g(\mathrm{d}z)$ is the distribution of $Z$ under the law of the particle
filter and $\alpha(z):=\int\left\{ 1\wedge\exp(z'-z)\right\} g(\mathrm{d}z')$
is the average acceptance probability from state $z$. Although not
emphasized notationally, both $g(z)$ and $\alpha(z)$ are functions
of $N$. Through an abuse of notation, we denote the invariant density
of the above chain with $\pi(z)=g\left(z\right)\exp\left(z\right)$.
Such a reparameterization has also been used in previous work to analyze
the PIMH and the related particle marginal MH algorithm \cite{pitt2012some,doucet2015biometrika}.

\section{Methodology\label{sec:Methodology}}

\subsection{Unbiased MCMC via couplings}\label{sec:couplingMCMC}

`Exact estimation' methods provide unbiased estimators
of expectations with respect to the stationary distribution of a Markov
chain using coupling techniques \cite{glynn2014exact,jacob2017unbiased}. In the following
we use these tools to couple PIMH kernels and estimate smoothing expectations,
after briefly reviewing the general approach.

Consider two $\mathsf{U}-$valued
Markov chains, $U:=(U^{(n)})_{n\ge0}$ and $\tilde{U}:=(\tilde{U}^{(n)})_{n\ge0}$,
each evolving marginally according to a kernel $K(u,\mathrm{d}u')$ with stationary
distribution $\lambda$ and initialized from $\eta$ so that $U^{(n)}\stackrel{d}{=}\tilde{U}^{(n)}$
for all $n\geq1$, where $\stackrel{d}{=}$ denotes equality in distribution.
We couple these two chains so that $U^{(n)}=\tilde{U}^{(n-1)}$
for $n\ge\tau$, where $\tau$ is an almost surely finite meeting
time. In this case, we see that for non-negative integers $k<t$ we
have the following telescoping-sum decomposition
\begin{align*}
\mathbb{E}[h(U^{(t)})] =&\mathbb{E}[h(U^{(k)})+\sum_{l=k+1}^{t}(h(U^{(l)})-h(U^{(l-1)}))] \\
=&\mathbb{E} [h(U^{(k)})+\sum_{l=k+1}^{t\land(\tau-1)} (h(U^{(l)})-h(\tilde{U}^{(l-1)}))].
\end{align*}
As a result, under integrability conditions which will be made precise
in the sequel, taking the limit as $t\rightarrow\infty$ suggests
an estimator $H_{k}(U,\tilde{U}):=h(U^{(k)})+\sum_{l=k+1}^{\tau-1}(h(U^{(l)})-h(\tilde{U}^{(l-1)}))$
with expectation $\lambda(h)$. A way to construct such chains considered
in \cite{glynn2014exact,jacob2017unbiased,heng2017unbiased,middleton2018unbiased}
relies on sampling independently $(U^{(0)},U^{(1)})\sim \eta(\mathrm{d}u_{0})K(u{}_{0},\mathrm{d}u_{1})$
and $\tilde{U}^{(0)}\sim \eta(\mathrm{d}u_{0})$, and then successively
sampling  $(U^{(n+1)},\tilde{U}^{(n)})$ from $\bar{K}((U^{(n)},\tilde{U}^{(n-1)}),(\cdot,\cdot))$ such that $U^{(n+1)}=\tilde{U}^{(n)}$
with positive probability, ensuring that both chains evolve marginally
according to $K$. Furthermore,
averaging $H_{l}(U,\tilde{U})$ over a range of values $l\in\{k,k+1,...,m\}$,
for some $m\geq k$, preserves unbiasedness, suggesting the following
`time-averaged' unbiased estimator of $\mu(h)$
\begin{align}
H_{k:m}  =\frac{1}{m-k+1}\sum_{l=k}^{m}h(U^{(l)})+\nonumber \enskip \qquad\qquad\qquad \\  \sum_{l=k+1}^{\tau-1}\min\left(1,\frac{l-k}{m-k+1}\right)(h(U^{(l)})-h(\tilde{U}^{(l-1)}))\nonumber \\
  :=\text{MCMC}_{k:m}(h)+\text{BC}_{k:m}(h).\label{eq:timeaveraged}
\end{align}
We view $\text{MCMC}_{k:m}(h)$ as the standard MCMC sample average
up to time $m$, discarding the first $k$ iterates as `burn-in'.
The second term, $\text{BC}_{k:m}(h)$, is the `bias correction' term
with $\text{BC}_{k:m}(h):=0$ if $\tau-1<k+1$. The estimator requires
only that two chains be simulated until meeting at time $\tau$, after
which, if $\tau<m$, only one chain must be simulated up to $m$.
By \cite[Proposition 3.1]{jacob2017unbiased}, the validity of the
resulting estimators is guaranteed under the following assumptions.
We discuss these assumptions for PIMH in Section \ref{sec:validity}.
\begin{assumption}
\label{assumption:marginal}Each chain is initialized marginally from
a distribution $\eta$, evolves according to a kernel $K$, and is
such that $\mathbb{E}[h(U^{(n)})]\to\lambda(h)$ as $n\to\infty$. Furthermore,
there exist constants $\eta>0$ and $D<\infty$ such that $\mathbb{E}[\left|h(U^{(n)})\right|{}^{2+\eta}]<D$
for all $n\geq0$.
\end{assumption}
\begin{assumption}
\label{assumption:meetingtime}The two chains are such that the meeting
time $\tau=\inf\{n\geq1:\ U^{(n)}=\tilde{U}^{(n-1)}\}$ satisfies $\mathbb{P}[\tau>n]\leq D\delta^{n}$,
for some constants $0<D<\infty$ and $\delta\in(0,1)$. The chains
stay together after meeting, i.e. $U^{(n)}=\tilde{U}^{(n-1)}$ for all
$n\geq\tau$.
\end{assumption}
\begin{prop}\cite{jacob2017unbiased}\ Under Assumptions \ref{assumption:marginal}
and \ref{assumption:meetingtime}, the estimator $H_{k:m}$ obtained by these coupled chains is unbiased, has finite variance and finite expected
cost.
\end{prop}
\subsection{Coupled PIMH}\label{sec:coupledPIMH}
To obtain unbiased estimators of $\pi(h)$, we use the framework detailed in Section \ref{sec:couplingMCMC} for
$K=P$ the transition kernel of PIMH and $\lambda=\bar{\pi}$ the corresponding invariant distribution of PIMH
defined on $\mathsf{U}:=\{1,...,N\}\times\mathcal{X}$. This requires introducing a coupling of PIMH kernels.
For IMH chains, a natural choice of initialization of the chain is from the proposal
distribution. We adopt this in the following, sampling both \textsf{$(X_{1:T}^{(0)},p_{N}^{(0)})\sim\text{PF}$}
and \textsf{$(\tilde{X}_{1:T}^{(0)},\tilde{p}_{N}^{(0)})\sim\text{PF}$}.
However, in contrast to previous constructions \cite{jacob2017unbiased,heng2017unbiased,middleton2018unbiased},
our method allows the chains to couple at time $\tau=1$
through re-using the initial value $(\tilde{X}_{1:T}^{(0)},\tilde{p}_{N}^{(0)})$
as a proposal for the first chain. %

We summarize the resulting procedure in Algorithm \ref{alg:couplepimh}. To simplify presentation of the algorithm, we use the convention $\alpha(\tilde{p}_{N}^{(-1)},p_{N}^{*}):=1$.
Finally, although the coupling scheme is framed in terms of PIMH, it is clear that this
algorithm and estimators apply equally to any IMH algorithm.

\begin{algorithm}[H]
\caption{Coupled PIMH\label{alg:couplepimh}}

\begin{enumerate}
\item \textsf{Sample $(X_{1:T}^{(0)},p_{N}^{(0)})\sim\text{PF}$, set $n=1$ and $\tau=\infty$.}
\item \textsf{While $n<\max(m,\tau)$}
\begin{enumerate}
\item \textsf{Sample $(X_{1:T}^{*},p_{N}^{*})\sim\text{PF}$.}
\item \textsf{Sample $\mathit{\mathfrak{u}}\sim\mathcal{U}\left[0,1\right]$.}
\item \textsf{If $\mathit{\mathfrak{u}}\le\alpha\left(p_{N}^{(n-1)},p_{N}^{*}\right)$
set \[(X_{1:T}^{(n)},p_{N}^{(n)})=(X_{1:T}^{*},p_{N}^{*})\]
else
$(X_{1:T}^{(n)},p_{N}^{(n)})=(X_{1:T}^{(n-1)},p_{N}^{(n-1)}).$}
\item \textsf{If $\mathit{\mathfrak{u}}\le\alpha(\tilde{p}_{N}^{(n-2)},p_{N}^{*})$
set
\[
(\tilde{X}_{1:T}^{(n-1)},\tilde{p}_{N}^{(n-1)})=(X_{1:T}^{*},p_{N}^{*})
\]
else $(\tilde{X}_{1:T}^{(n-1)},\tilde{p}_{N}^{(n-1)})=(\tilde{X}_{1:T}^{(n-2)},\tilde{p}_{N}^{(n-2)}).$}
\item \textsf{If $\mathit{\mathfrak{u}}\le\alpha\left(p_{N}^{(n-1)},p_{N}^{*}\right)\wedge\alpha(\tilde{p}_{N}^{(n-2)},p_{N}^{*})$
 set $\tau=n$.}
\item \textsf{Set  $n \leftarrow n+1$.}
\end{enumerate}
\item \textsf{Return $H_{k:m}$ as in (\ref{eq:timeaveraged})}.
\end{enumerate}
\end{algorithm}

At each iteration $n$ of Algorithm \ref{alg:couplepimh} the proposal
$(X_{1:T}^{*},p_{N}^{*})$ can be accepted by one, both or neither
chains with positive probability. The meeting time $\tau$
corresponds to the first time the proposal $(X_{1:T}^{*},p_{N}^{*})$
is accepted by both chains. We can also return another unbiased estimator $\bar{H}_{k:m}$ of the form (\ref{eq:timeaveraged})
with $h(X_{1:T})$ replaced by $\pi_{N}(h)$ as
we have previously seen that $\bar{\pi}(\pi_{N}(h))=\pi(h)$.
The Rao-Blackwellised estimator $\bar{H}_{k:m}$ will typically outperform significantly $H_{k:m}$ when we are interested in
smoothing expectations of functions of states close to $T$. For functions of states close to the origin, e.g. $h(x_{1:T})=x_{1}$, we have
$\bar{H}_{k:m} = H_{k:m}$ with high probability if $N$ is moderate because of the particle degeneracy problem \cite{doucet2009tutorial,jacob2015path,kantas2015particle}.
 This is illustrated in Appendix \ref{subsec:raoblackstochkinetic}.

\subsection{Validity and meeting times\label{sec:validity}}
The validity of our unbiased estimators is ensured if the following weak assumptions are satisfied.
\begin{assumption}
\label{assumption:marginaldistributions}There exist constants $\eta>0$
and $D<\infty$ such that $\mathbb{E}\left[|h(X_{1:T}^{(n)})|{}^{2+\eta}\right]<D$
for all $n\geq0$.
\end{assumption}
\begin{assumption}
\label{assumption:marginaldistributionsjoint}There exist constants $\eta>0$
and $D<\infty$ such that
$\mathbb{E}\left[|\sum_{i=1}^N W_T^{i,(n)} h(X_{1:T}^{i,(n)})|{}^{2+\eta}\right]<D$
for all $n\geq0$.
\end{assumption}
\begin{assumption}
\label{assumption:pf}The resampling scheme is unbiased and there exist
finite constants $(\bar{w_{t}})_{t=1}^{T}$ such that $\sup_{x\in\mathsf{X}}$
$w_{1}(x)\le\bar{w}_{1}$ and $\sup_{(x,x')\in\mathsf{X}\times\mathsf{X}}$
$w_{t}(x,x')\le\bar{w}_{t}$ for $t\in\{2,...,T\}$.
\end{assumption}

\begin{prop}
\label{thm:validity} Under Assumptions \ref{assumption:marginaldistributions}
and \ref{assumption:pf}, resp. Assumptions \ref{assumption:marginaldistributionsjoint}
and \ref{assumption:pf}, the estimator $H_{k:m}$, resp. $\bar{H}_{k:m}$, of $\pi(h)$ obtained from Algorithm
\ref{alg:couplepimh} is unbiased and has finite variance and finite expected
cost.
\end{prop}

To establish the result for $H_{k:m}$, note that Assumption \ref{assumption:marginaldistributions} and
our construction implies Assumption
\ref{assumption:marginal} for $\lambda=\bar{\pi}$. Assumption \ref{assumption:pf}
provides verifiable and sufficient conditions to ensure uniform ergodicity of PIMH \cite[~Theorem 3]{andrieu2010particle}.
Hence it follows from \cite[Proposition 3.4]{jacob2017unbiased} that the geometric bound on the tails of $\tau$ is satisfied.
Additionally, Algorithm \ref{alg:couplepimh} ensures that the chains stay together
for all $n\ge\tau$ so Assumption \ref{assumption:meetingtime} is satisfied. A similar reasoning provides the result for $\bar{H}_{k:m}$.

We have the following precise description of the distribution of the
meeting time $\tau$ for Algorithm \ref{alg:couplepimh}. Let $\text{Geo}(\gamma)$ denote the geometric distribution on the strictly positive integers
with success probability $\gamma$.
\begin{prop}
\label{prop:mtdist}
The meeting time satisfies $\tau|Z^{(0)}\sim\emph{\text{Geo}}(\alpha(Z^{(0)}))$ and $\mathbb{P}[\tau=1]\geq\frac{1}{2}$. Additionally, we have
$\lim_{N\rightarrow\infty}\mathbb{P}[\tau=1]=1$ under Assumption \ref{assumption:pf}.
\end{prop}
We recall here that $\alpha(z)$ is the average acceptance probability from state $z$. The geometric distribution result follows from Proposition \ref{prop:stochasticdominatingcoupling} in Appendix. The rest of the proposition follows from
\begin{align*}
\mathbb{P}[\tau=1]&=\int\alpha(z)g(\mathrm{d}z)\\
&=\iint\left\{1\wedge\exp(z'-z)\right\}g(\mathrm{d}z)g(\mathrm{d}z').
\end{align*}
It entails trivially that $\mathbb{P}[\tau=1]\geq\frac{1}{2}$. Noting that $\lim_{N\rightarrow\infty}Z=0$
a.s. under $g$ (which is dependent on $N$) under Assumption \ref{assumption:pf}, see e.g. \cite{Del_Moral_2004}, we obtain $\lim_{N\rightarrow\infty}\mathbb{P}[\tau=1]=1$ by dominated convergence, thus $\lim_{N\rightarrow\infty}\text{BC}_{k:m}(h)=0$.

\subsection{Unbiased filtering}

A PF generates estimates $p_{N}(y_{1:t})$ of $p(y_{1:t})$ at all times $t=1,\ldots,T$. These estimates can be used to perform one step
of coupled PIMH for $T$ pairs of Markov chains, each pair corresponding to
one of the smoothing distributions $p(x_{1:t}|y_{1:t})$. This requires some additional bookkeeping
to keep track of the meeting times and unbiased estimators associated with each $p(x_{1:t}|y_{1:t})$ but can be used to
unbiasedly estimate expectations with respect to all filtering distributions $p(x_t|y_{1:t})$. This was not directly feasible with coupled i-CPF in \cite{jacob2017smoothing}.
In particular, this allows us to estimate unbiasedly the predictive likelihood terms $p(y_t|y_{1:t-1})$
which can be used for a goodness-of-fit test \cite{gerlach1999diagnostics}.

\section{Analysis \label{sec:guidelines}}

\subsection{Meeting time: large sample approximation\label{subsec:CLT}}

We investigate here the distribution of the meeting time in the large
sample regime, i.e. in the interesting scenarios where $T$ is large. Under strong mixing assumptions, \cite{berard2014lognormal}
showed that letting $\frac{T}{N}=\gamma$ for some $\gamma>0$ then the following Central Limit Theorem (CLT) holds:
$Z\stackrel{d}{\rightarrow}\mathcal{N}\left(-\frac{1}{2}\sigma^{2},\sigma^{2}\right)$
as $T\rightarrow\infty$ where $\sigma^{2}=\gamma\bar{\sigma}^{2}$
for some $\bar{\sigma}^{2}>0$. Empirically, this CLT appears to hold for many realistic models not satisfying these strong mixing assumptions \cite{pitt2012some,doucet2015biometrika}.
Under this CLT, using similar regularity conditions as in \cite{schmon2018large},
it can be shown that the Markov kernel $Q$ defined in (\ref{eq:PforZ})
converges in some suitable sense towards the Markov kernel
$Q_{\sigma}$ given by
\begin{equation*}
\left\{1\wedge\exp(z'-z)\right\}g_{\sigma}(z')\mathrm{d}z'+\left\{ 1-\alpha_{\sigma}(z)\right\} \delta_{z}(\mathrm{d}z'),\label{eq:approxP}
\end{equation*}
where $g_{\sigma}(z)=\mathcal{N}\left(z;-\frac{1}{2}\gamma\bar{\sigma}^{2},\gamma\bar{\sigma}^{2}\right)$
and $\alpha_{\sigma}(z):=\int\left\{1\wedge\exp(z'-z)\right\}g_{\sigma}(z')\mathrm{d}z'$
denotes the average acceptance probability in $z$. For this limiting kernel\footnote{Note that $Q_{\sigma}$ is not uniformly ergodic whereas $Q$ is under Assumption \ref{assumption:pf}.},
the following result holds.
\begin{prop}
\label{prop:approxchain} \cite[Corollary 3]{doucet2015biometrika} The
invariant density of $Q_{\sigma}$ is $\pi_{\sigma}(z)=\mathcal{N}\left(z;\frac{1}{2}\sigma^{2},\sigma^{2}\right)$
and its average acceptance probability is given by
\[
\alpha_{\sigma}(z):=1-\Phi\left(\frac{z+\frac{\sigma^{2}}{2}}{\sigma}\right)+e^{-z}\Phi\left(\frac{z-\frac{\sigma^{2}}{2}}{\sigma}\right),
\]
where $\Phi$ denotes the cumulative distribution function of the
standard Normal distribution.
\end{prop}
Under this large sample approximation, the probability $\mathbb{P}[\tau=n]$ can be written as
\begin{equation}
\mathbb{E}\left[\mathbb{P}[\tau=n|Z^{(0)}]\right]=\int\alpha_{\sigma}(z)(1-\alpha_{\sigma}(z))^{n-1}g_{\sigma}(z)\mathrm{d}z,\label{eq:limptau}
\end{equation}
thus $\mathbb{P}[\tau=1]=\frac{1}{2}\left\{1+\exp(\sigma^{2})\text{Erfc}(\sigma)\right\}$, Erfc denoting the complementary error function.
Similarly we see that the expected meeting time is given by
$\mathbb{E}[\tau]=\mathbb{E}_{g_{\sigma}}\left[\alpha_{\sigma}(Z)^{-1}\right]$.
This allows us to approximate numerically expectations, quantiles
and probabilities of $\tau$ as a function of $\sigma$, the standard
deviation of $\log p_{N}(y_{1:T})$ under the law of the particle filter.
Figure \ref{fig:psummary} displays $\mathbb{P}[\tau=1]$ and $\mathbb{E}[\tau]$ as a function of $\sigma$. We see that $\mathbb{P}[\tau=1]=0.71$ for $\sigma=1$ rising to $\mathbb{P}[\tau=1]=0.95$ for $\sigma=0.1$. The expected meeting time $\mathbb{E}[\tau]$ is also the expected number of iterations to return an unbiased estimator for $m=0$ and grows relatively benignly with $\sigma$.

\begin{figure}[ht]
\begin{centering}
\includegraphics[scale=0.4]{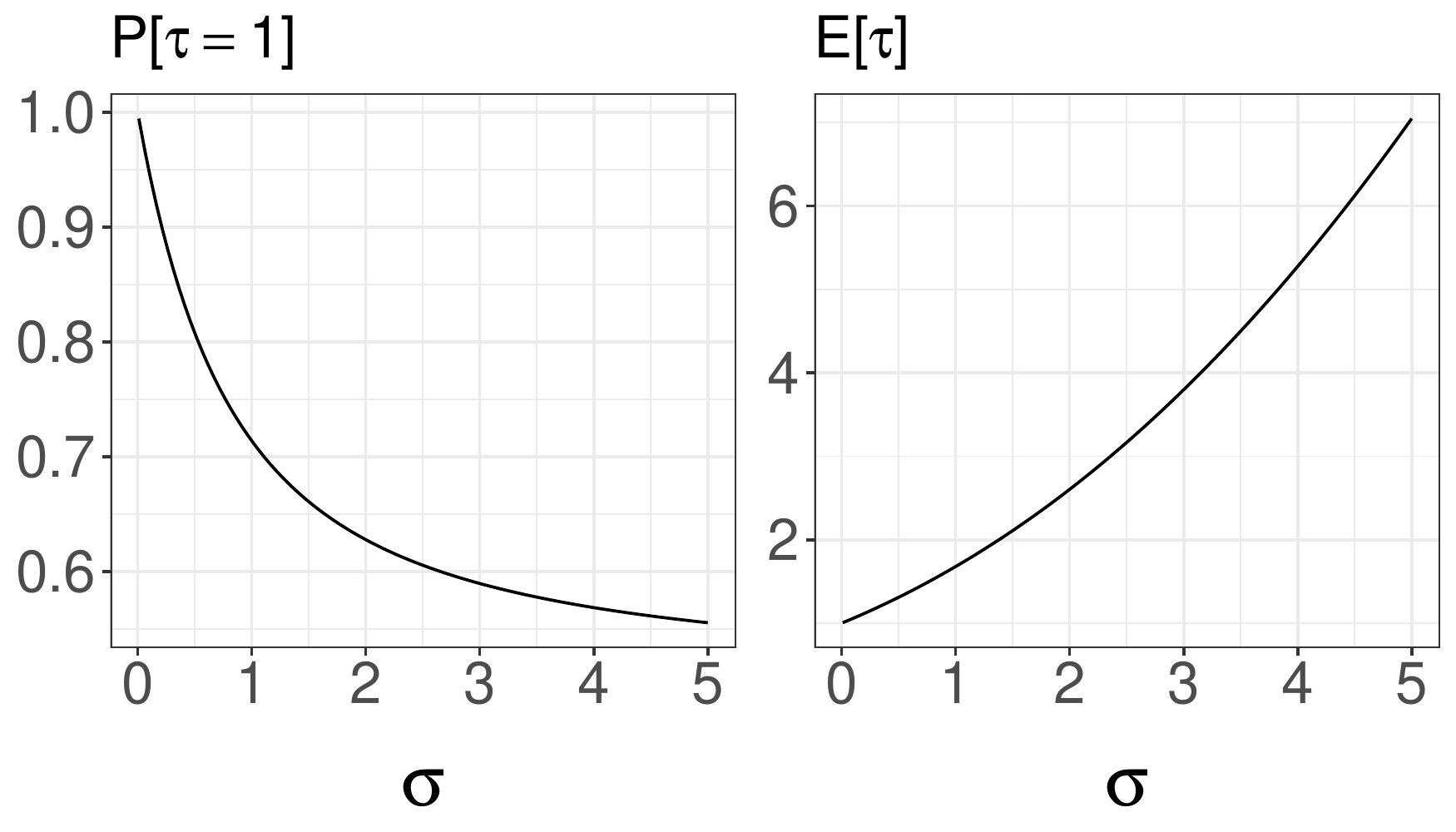}

\caption{$\mathbb{P}[\tau=1]$ (left) and $\mathbb{E}[\tau]$ (right) as a function of the standard deviation $\sigma$  of $\log {p}_N(y_{1:T})$. \label{fig:psummary}
}
\end{centering}
\end{figure}

We compare the distribution of the meeting times with its large sample approximation (\ref{eq:limptau}) on a
stationary auto-regressive (AR) model $X_{t}\sim\mathcal{N}(aX_{t-1},1)$
and $Y_{t}\sim\mathcal{N}(X_{t},\sigma_{y}^{2})$ with $a=0.5$ and
$\sigma_{y}^{2}=10$ on a simulated dataset of $T=100$. Estimates
of $\mathbb{P}[\tau\ge n]$ were obtained empirically for a range
of values of $N$ and compared with the predicted values based on
the estimated variance $\sigma^{2}$ of $\log p_{N}(y_{1:T})$
using $10^{5}$ runs of coupled PIMH.
Varying $N$ between 10 and 110, $\sigma^{2}$ was between 0.2 and 3.0 in this range.
The tail probabilities of the meeting time over this range are shown
in Figure \ref{fig:lgssmmt}.
Confidence intervals for the estimates of $\mathbb{P}[\tau\ge n]$
are shown in Figure \ref{fig:lgssmmt} with error bars indicating
$\pm2$ standard deviations. We see that there is a satisfactory agreement, with
predicted tail probabilities closely matching confidence intervals
for each value, and larger values of $N$ leading as expected to shorter
meeting times on average.%

\begin{figure}[ht]
\begin{centering}

\begin{centering}
\includegraphics[scale=0.4]{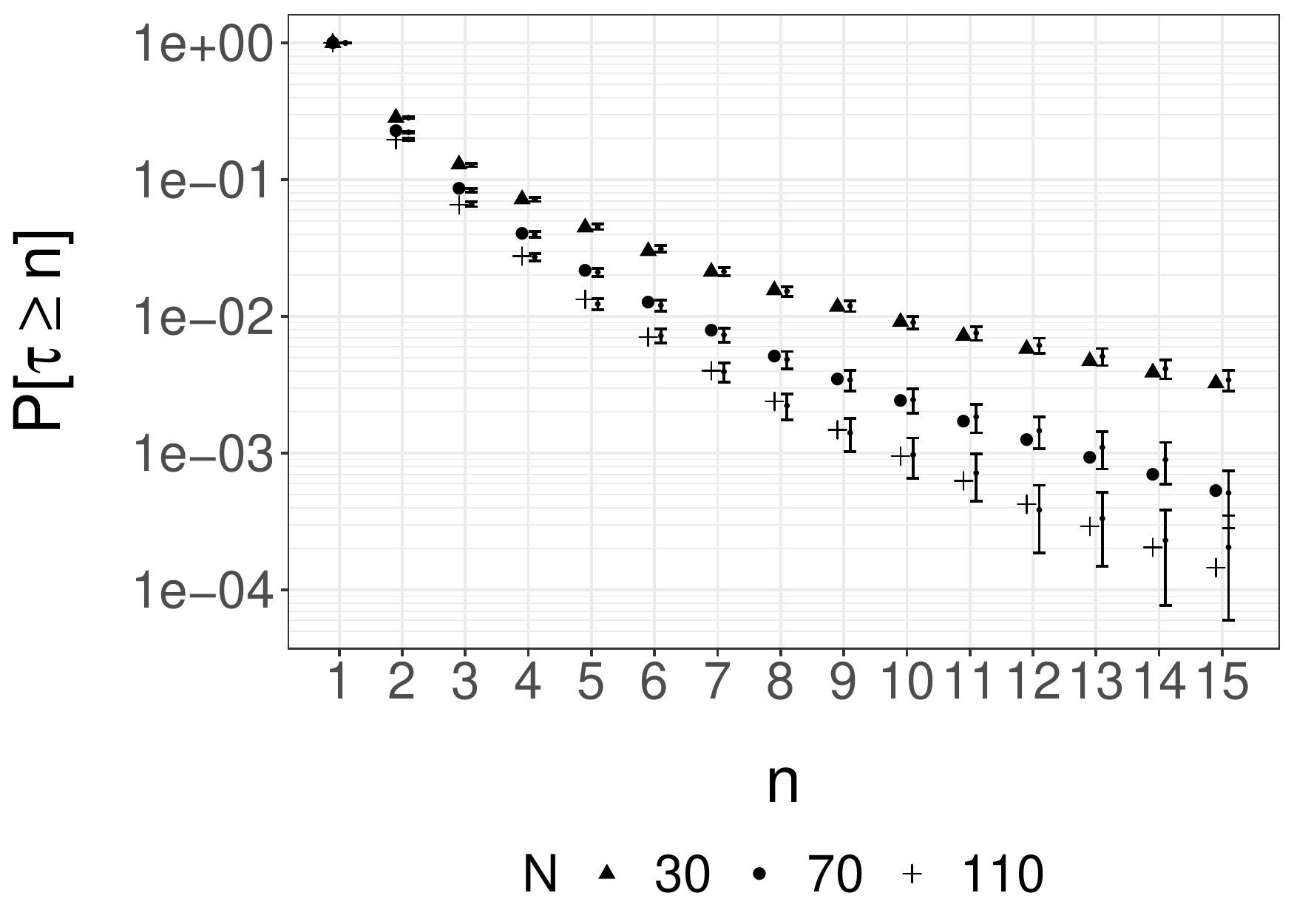}
\par\end{centering}
\par\end{centering}
\caption{Empirical estimates of $\mathbb{P}[\tau\ge n]$ (with $\pm2$ standard errors) and comparison with the large sample approximation for toy example.\label{fig:lgssmmt}}
\end{figure}

\subsection{On the selection of $N$ for large $m$}\label{guide}
We provide here a heuristic for selecting $N$ to minimize the variance of $H_{k:m}$ at fixed
computational budget when both $m$ and $T$ are large. We will minimize the computational inefficiency defined by
\[
\text{C}[H]:=\mathbb{V}[H_{k:m}]\times N
\]
as the running time is proportional to $N$. For $m$ sufficiently large, we expect that $\tau <m$ with
very high probability and so the time to obtain a single unbiased
estimator is $m$. We note that increasing
$N$ typically leads to a decreasing asymptotic variance of the ergodic averages
associated with the PIMH chain and from Proposition \ref{prop:mtdist}
a reduction in the bias correction, however at the cost of more computation.
For large $m$, we also expect that the dominant term of $\mathbb{V}[H_{k:m}]$
will arise from $\text{MCMC}_{k:m}(h)$ and will be essentially the asymptotic variance of $h$ given
by $\mathbb{V}_{\pi}[h]  \text{IF}(h)$ divided by $(m-k+1)$, where $\text{IF}(h)$ is the Integrated Autocorrelation Time (IACT) of $h$ for the PIMH kernel.
For $N$ sufficiently large, we expect that $X_{1:T}$ and $p_{N}$ are approximately
independent under the PF proposal. By a reasoning similar to the proof of \cite[Lemma 4]{pitt2012some}, $\text{IF}(h)$ will then be approximately proportional to $\text{IF}(\sigma)$ defined in Eq.\ (11) in \cite{pitt2012some}. This is illustrated in Appendix \ref{sec:inefficiency} where we plot $\text{IF}(h)$ for a variety of test functions for a range of $N$  against $\text{IF}(\sigma)$ over the corresponding range of $\sigma$ and show that they are indeed approximately proportional.
Minimizing $\text{IF}[H]$ w.r.t.\ $N$ is then approximately equivalent to minimizing $\text{IF}(\sigma)/\sigma^{2}$ as $\sigma^{2}$ is typically
inversely proportional to $N$. This minimization has already been carried out in \cite{pitt2012some} where it was found
that the minimizing argument is $\sigma=0.92$. Practically, this means that one should select $N$ to ensure that
the standard deviation of $\log p_{N}(y_{1:T})$ is equal approximately to this value. The resulting value of $N$ is expected to be close to the value of $N$ minimizing $\text{C}[H]$, which is approximately proportional to $\text{IF}(h)/\sigma^{2}$. This is verified in Figure \ref{fig:toyineffopt} on the AR example of Section \ref{subsec:CLT}. Note that these guidelines do not apply to $\bar{H}_{k:m}$ for $h$ a function of states close to $T$ as it is not true that
$\pi_{N}(h)$ and $p_{N}$ are approximately independent.

\begin{figure}[ht]
\begin{centering}
\includegraphics[scale=0.4]{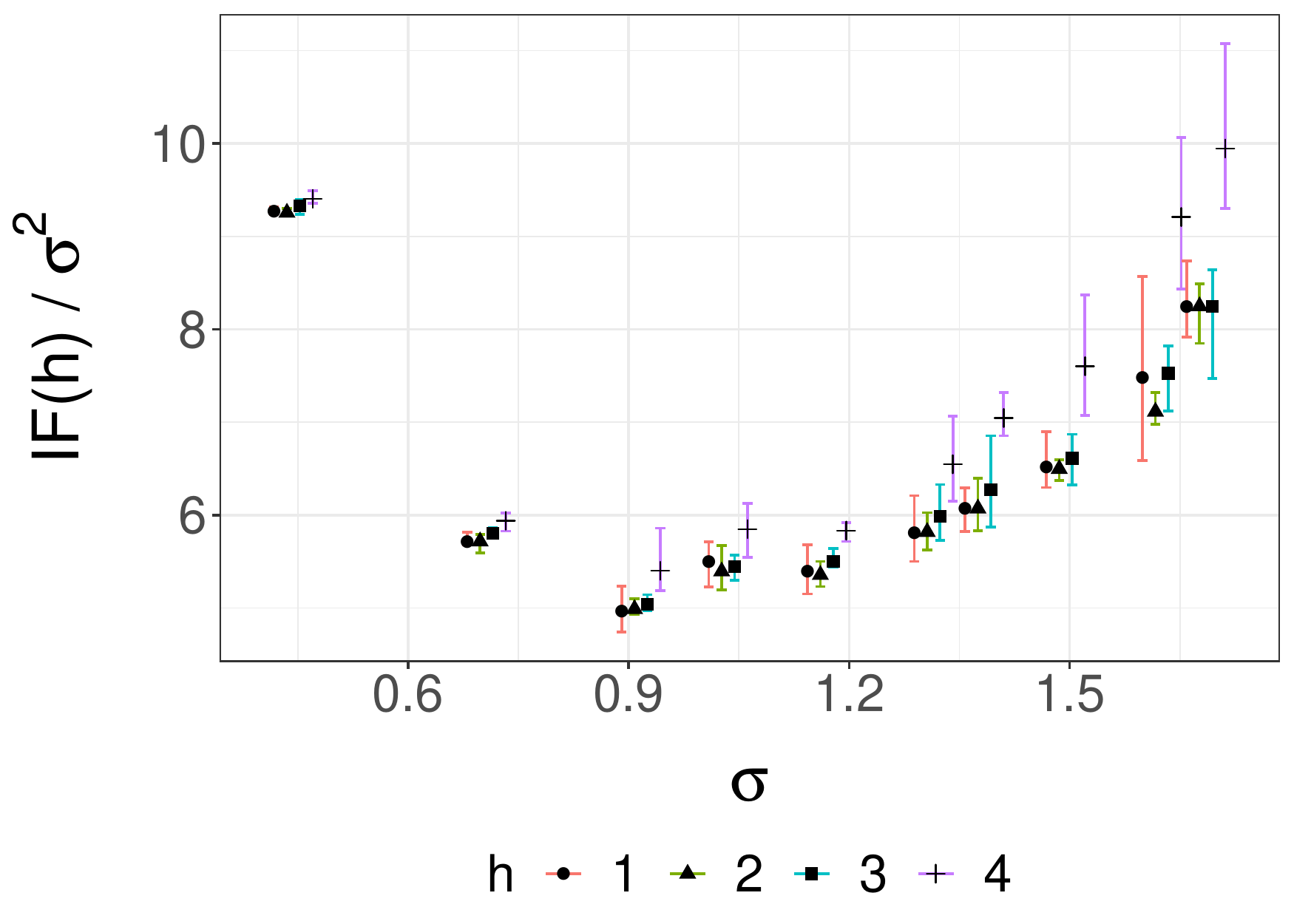}
\caption{ $\text{IF}(h_i)/\sigma^{2}$ for $h_1=x_1$,  $h_2=x_T$,  $h_3=\sum_t x_t$ and  $h_4=\sum_t x_t^2$ as a function of $\sigma^{2}$.\label{fig:toyineffopt}
}
\end{centering}
\end{figure}

\section{Numerical experiments\label{sec:Numerical-experiments}}

We apply the methodology to two models where the transition density
of the latent process is analytically intractable but simulation from it is possible.
However, this involves sampling a random number of random variables.
The unbiased smoother proposed in \cite{jacob2017smoothing} relies on common random numbers
and it is unclear how one could implement it in this context. The coupled conditional backward sampling PF scheme proposed in \cite{lee2018coupled}
does not apply as we cannot evaluate the transition density pointwise.
In both scenarios, we use the bootstrap PF with multinomial resampling, that is $q_{1}(x_{1})=\mu(x_{1})$ and $q_{t}(x_{t}|x_{t-1})=f(x_{t}|x_{t-1})$, and Assumption \ref{assumption:pf}
is satisfied.

\subsection{Stochastic kinetic model}
We consider a stochastic kinetic model represented by a jump Markov process introduced in \cite{golightly2005bayesian}.
Such models describe a system of chemical reactions in continuous
time, with a reaction occurring under the collision of two species
at random times. The discrete number of each species describes the
state with jumps representing the change in a particular species.
The discrete valued state $(X_{t,q})_{t\ge0,1\le q\le Q}$ comprises
a $Q$-vector of species at each time, where one of $R$ reactions
may occur at any random time, given by hazard functions $f_{r}$ for
$r\in\{1,...,R\}$. The effect of such a reaction is described by
a stoichiometry matrix $S$, where the instantaneous change in the
number of species $q$ for a certain reaction $r$ out of a possible
$R$ different reactions is encoded in element $S_{q,r}$. For the prokaryotic autoregulation model and parameterisation considered in \cite{golightly2017efficient},
the state is a four dimensional vector evolving according to $8$
possible reactions, for which the stoichiometry matrix is given by
\begin{align*}
S  &=\left(\begin{array}{cccccccc}
0 & 0 & 1 & 0 & 0 & 0 & -1 & 0\\
0 & 0 & 0 & 1 & -2 & 2 & 0 & -1\\
-1 & 1 & 0 & 0 & 1 & -1 & 0 & 0\\
-1 & 1 & 0 & 0 & 0 & 0 & 0 & 0
\end{array}\right),\\
f(X,c) &=(c_{1}X_{4}X_{3},c_{2}(k-X_{4}),c_{3}X_{4},c_{4}X_{1},\\&c_{5}X_{2}(X_{2}-1)/2,c_{6}X_{3},c_{7}X_{1},c_{8}X_{2})'.
\end{align*}
Coefficients $c_{1:8}$ and $k$ are parameters of the model, given
by $c=(0.1,0.7,0.35,0.2,0.1,0.9,0.3,0.1)$ and $k=10$. We collect $T=100$ noisy observations of the latent process at regular intervals of length $\Delta=0.1$, i.e.
\[
Y_{t}=\left(\begin{array}{cccc}
1 & 0 & 0 & 0\\
0 & 1 & 2 & 0
\end{array}\right)X_{\Delta t}+\epsilon_{t},\quad \epsilon_{t}\stackrel{i.i.d}{\sim}\mathcal{N}(0,I_{2}).
\]
We fix the initial condition $X_{0}=(8,8,8,5)$.

To simulate synthetic data and to run the bootstrap PF, we sample the latent process $X_t$ using Gillespie's direct method \cite{gillespie1977exact},
whereby the time to the next event is exponential with rate $\sum_{r=1}^{R}f_{r}(X,c)$
and reaction $r$ occurs with probability $f_{r}(X,c)/\sum_{r=1}^{R}f_{r}(X,c)$.
The estimated survival probabilities
of the meeting time $\tau$, with $\pm2$ standard errors, computed using 500 independent runs of coupled PIMH are plotted in Figure \ref{fig:stochkin},
along with the probabilities obtained from the large sample approximation
in Section \ref{subsec:CLT},
showing good agreement between the two.
In Figure \ref{fig:sksderr} we display the unbiased smoothing estimators obtained by averaging the unbiased estimators obtained over 500 independent runs for $N=1,000,k=m=0$ and the corresponding confidence intervals. Alternative choices of $k,m$ can lead to improved performance at fixed computational budget \cite{jacob2017unbiased,middleton2018unbiased}.

\begin{figure}[ht]
\begin{centering}

\begin{centering}
\includegraphics[scale=0.4]{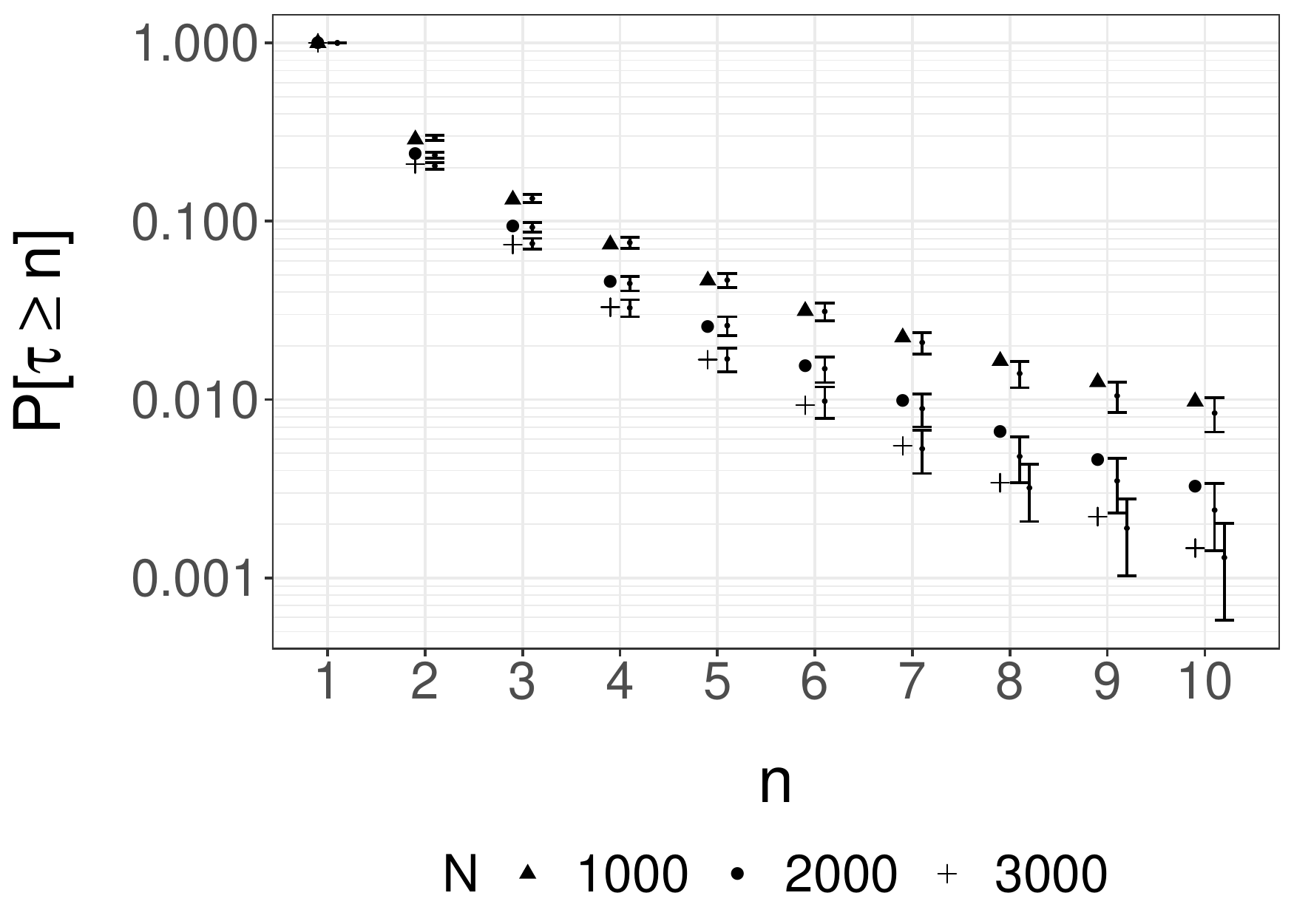}
\par\end{centering}
\label{fig:predictedtaugsss-1}
\par\end{centering}
\caption{Empirical estimates of $\mathbb{P}[\tau\ge n]$ (with $\pm2$ standard errors) and comparison with those implied by the large sample approximation for stochastic kinetic model.}
\label{fig:stochkin}
\end{figure}

\begin{figure}[ht]
\begin{centering}

\begin{centering}
\includegraphics[scale=0.4]{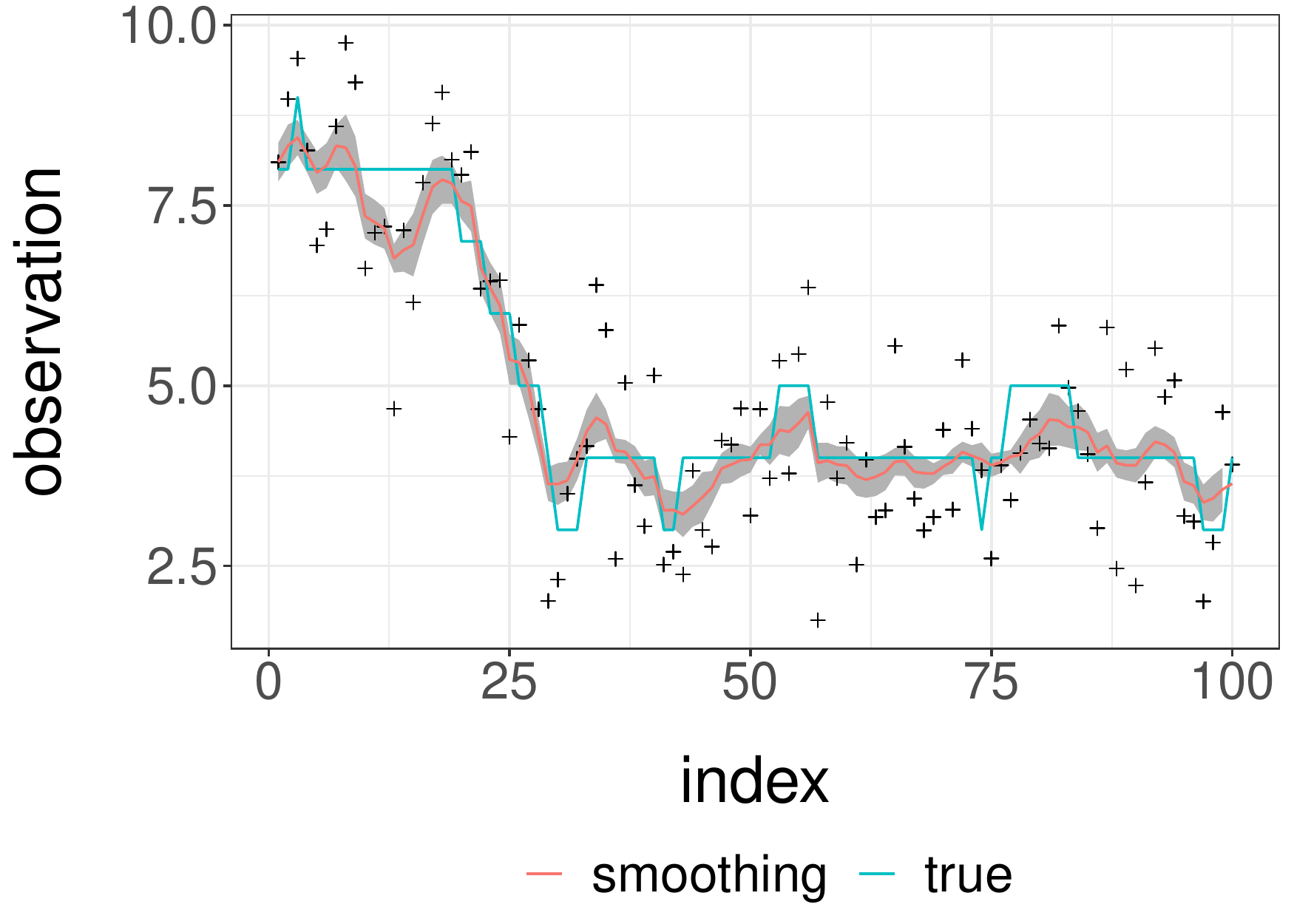}
\par\end{centering}
\label{fig:predictedtaugsss-1}
\par\end{centering}
\caption{Unbiased estimates of  $\mathbb{E}(X_{1,t}\vert y_{1:T})$ (red) (with $\pm3$ pointwise standard errors) and $X_{1,t}$ (blue) for Markov jump process.}
\label{fig:sksderr}
\end{figure}

\subsection{L\'evy-driven stochastic volatility}

Introduced in \cite{barndorff2001non}, L\'evy-driven stochastic volatility
models provide a flexible model for the log-returns of a financial
asset. Letting $(Y_{t})$ denote the log-return process, we have
\[
Y_{t}=\mu+\beta V_{t}+V_{t}^{1/2}\epsilon_{t},\qquad\epsilon_{t}\stackrel{i.i.d.}{\sim}\mathcal{N}(0,1),
\]
where $V_{t}$ (termed the `actual volatility') is treated as a stationary
stochastic process. The latent state comprises the pair
of actual and spot volatility $X_{t}=(V_{t},W_{t})$. In the terminology of \cite{barndorff2002econometric},
the integrated volatility is the integral of the spot volatility and
the actual volatility is an increment of the integrated volatility
over some unit time. Initializing $Z_{0}\sim\Gamma(\xi^{2}/\omega^{2},\xi^{2}/\omega^{2})$, the process evolves through sampling the following
random variables and recursing the state $X_{t}=(V_{t},W_{t})$ according
to
\begin{align*}
K &\sim\text{Poisson}(\lambda\xi^{2}/\omega^{2}),\quad C_{1:K}  \stackrel{i.i.d.}{\sim}\mathcal{U}[t-1,t],  \\
E_{1:K}&\stackrel{i.i.d.}{\sim}\text{Exp}(\xi/\omega^{2}),\quad
W_{t}=e^{-\lambda}W_{t-1}+\sum_{j=1}^{K}e^{-\lambda(t-C_{j})}E_{j}, \\
V_{t}&=\frac{1}{\lambda}(W_{t-1}-W_{t}+\sum_{j=1}^{K}E_{j}).
\end{align*}
In particular, conditionally on $X_{t-1}$, simulation of
$X_{t}$ requires a random number of random numbers sampled
at each iteration. The parameters $(\xi,\omega^{2})$ denote the stationary
mean and variance of the spot volatility respectively, $\lambda$
describes the exponential decay of autocorrelations, $\beta$ denotes
the risk premium for excess volatility and $\mu$ the drift of the
log-return. In the following we perform unbiased smoothing of $X_{t}$
using $T=500$ data from the S\&P 500 index used in \cite{chopin2013smc2}. A summary
of the data and parameter inference is included in Appendix \ref{subsec:datasummary}.

\begin{figure}
%
\subfloat[Empirical ($\pm2$ standard errors) and predicted tail probabilities implied by large-sample approximation ${\mathbb{P}[\tau\ge n]}$.]
{\begin{centering}
\includegraphics[scale=0.4]{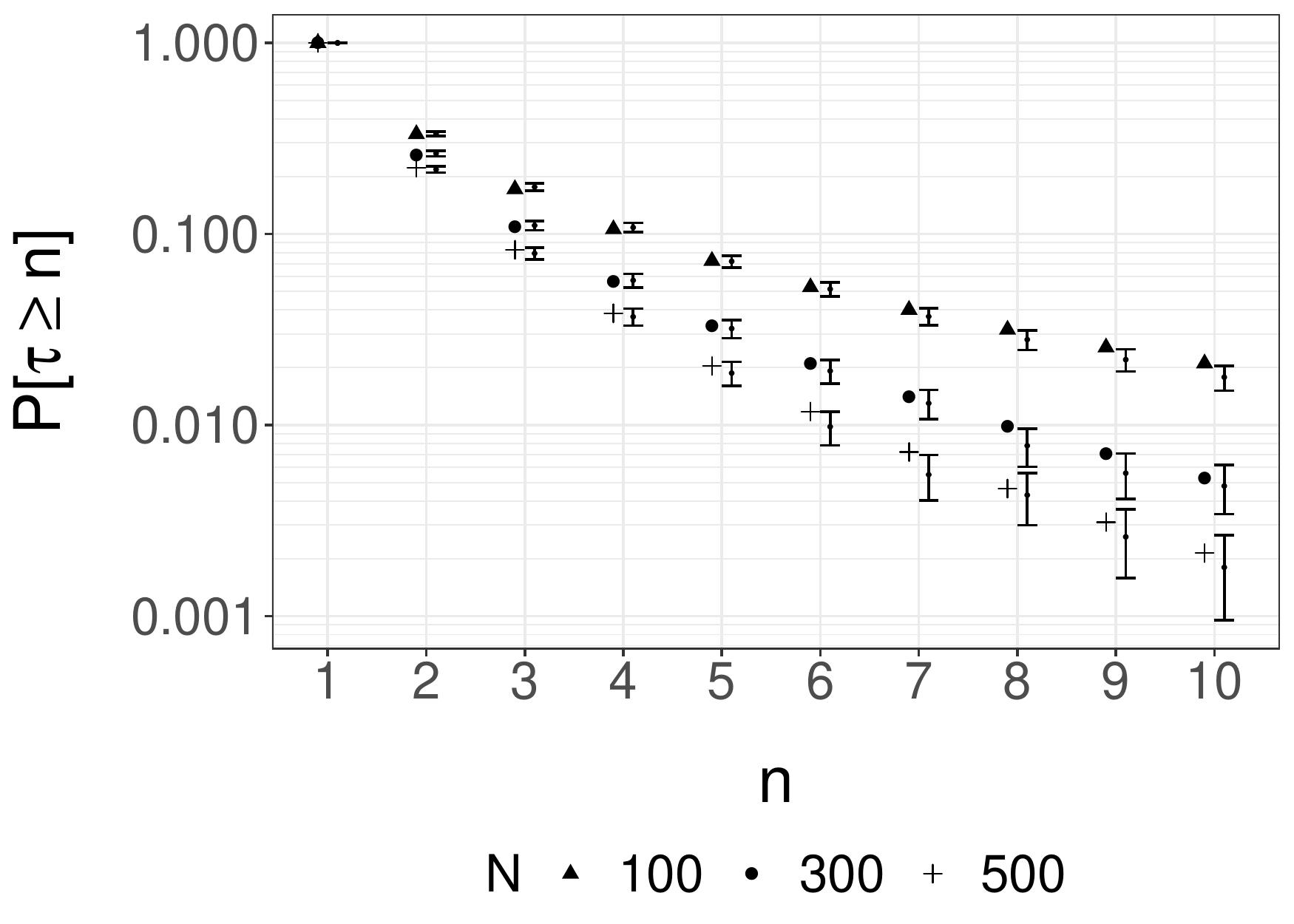}
\par\end{centering}
\centering{}\label{fig:mtsv}}

\begin{centering}
\subfloat[{Estimates of $(m-k+1)\mathbb{V}[H_{k:m}]/(\mathbb{V}_\pi[h] \sigma^{2})$ ($\pm3$ standard errors) computed using 10,000 runs.}]{\begin{centering}
\includegraphics[scale=0.4]{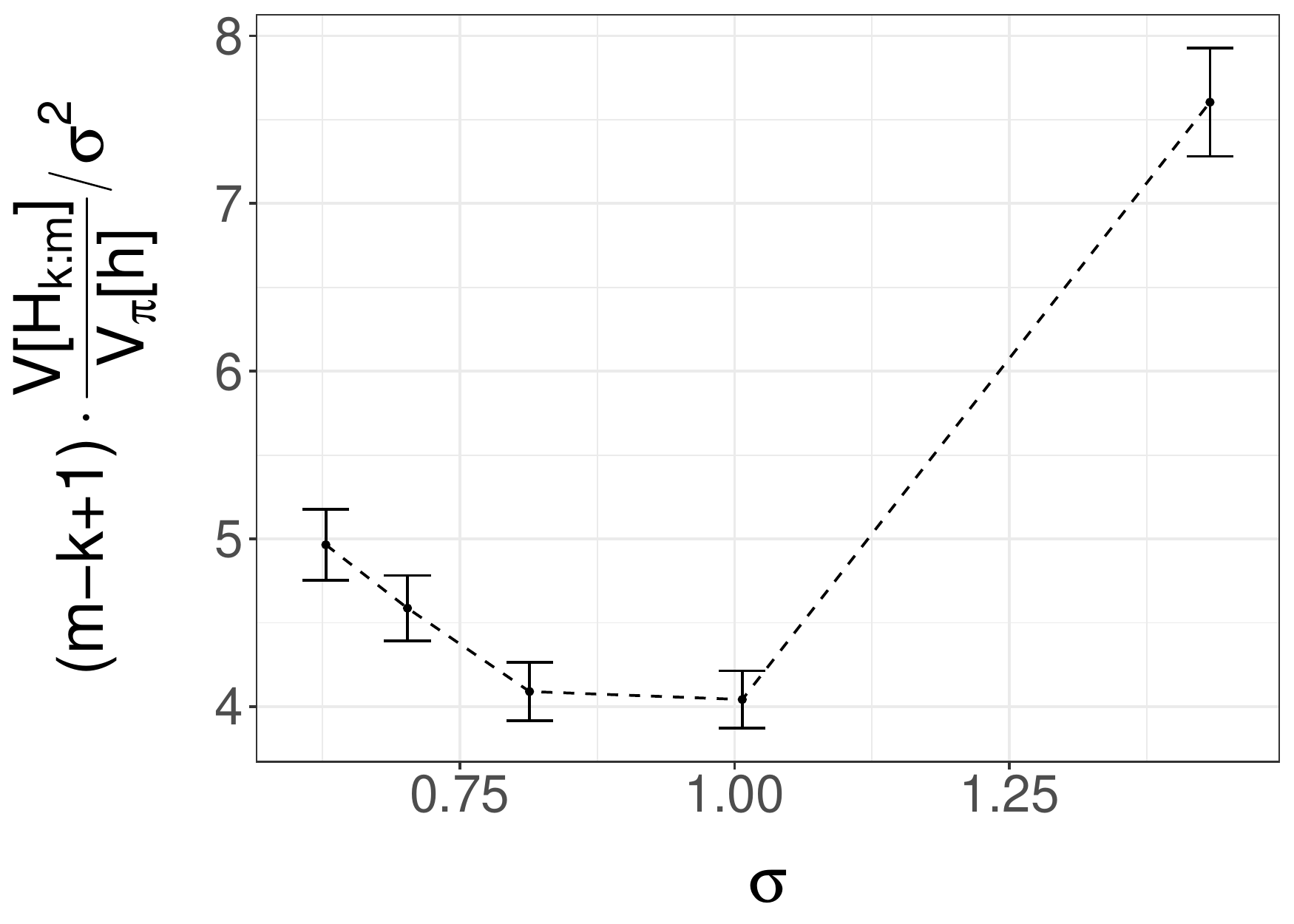}
\par\end{centering}
\centering{}\label{fig:svineff}}

\subfloat[Unbiased estimates of $\mathbb{E}(W_{t}\vert y_{1:T})$ of the S\&P 500
over 2005-2007 and confidence intervals at $\pm3$ standard errors. ]{\begin{centering}
\includegraphics[scale=0.4]{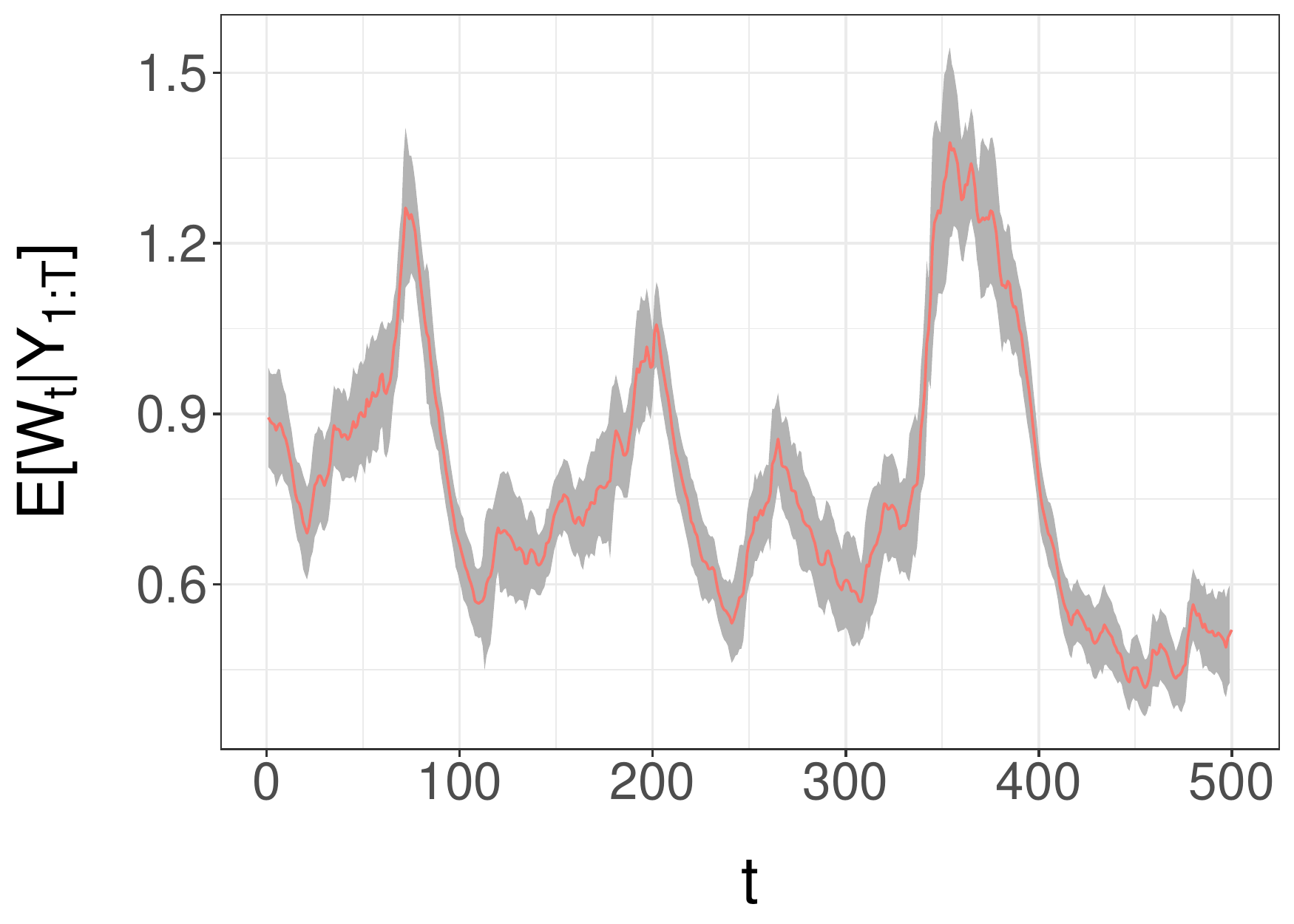}
\par\end{centering}
\centering{}\label{fig:mtsv-1}}
\par\end{centering}
\centering{}

\caption{Distribution of meeting times (top), computational inefficiency (bottom) and smoothing state estimators for L\'evy-driven stochastic volatility applied to real data (bottom).}

\label{fig:svmt}
\end{figure}

The empirical distributions of the meeting time for for $N\in\{100,...,500\}$ obtained using $1,000$ runs
are shown in Figure \ref{fig:svmt}, again showing good agreement with the large-sample approximation.
We then obtain unbiased estimators $H_{k:m}$ for $h(x_{1:T})=\sum_{t=1}^T v_{t}$ using $k=20,m=512$ over this grid of $N$.
Figure \ref{fig:svineff} plots $(m-k+1)\mathbb{V}[H_{k:m}]/(\mathbb{V}_\pi[h] \sigma^{2}$). We expect this function to be close to $\text{IF}(h)/\sigma^{2}$ as $m\rightarrow\infty$ and $\text{IF}(h)/\sigma^{2}$ to be minimized around 0.92. The experiments are consistent with this result. Figure \ref{fig:mtsv-1} presents the unbiased estimates of the spot volatility
$W_{t}$ obtained by averaging the unbiased estimates obtained over 1,000 runs and the corresponding confidence intervals for $k=m=0$ and $N=100$. Different choices of $k,m$ could
lead to improved performance at fixed computational budget.

\section{Discussion}

We have introduced a simple approach to perform unbiased smoothing
in state-space models and we have provided guidance on the choice of tuning parameters through
appealing to a large sample approximation.

We have established the validity of the estimators when the incremental weights
are bounded (Assumption \ref{assumption:pf}) which ensures uniform ergodicity of the PIMH.
Rejection sampling is possible under a similar assumption and would provide
exact samples from the smoothing distribution. However the expected number of trials before acceptance of such
a rejection scheme increases typically exponentially fast with $T$.
If we are only interested in obtaining unbiased smoothing estimators and if the CLT discussed
in Section \ref{subsec:CLT} holds, we expect our coupling scheme to only require increasing $N$ linearly with $T$ to control $\sigma$ and thus the corresponding expectation of the meeting time.

Finally, the scheme proposed here can be extended to obtain unbiased estimators of expectations
with respect to any posterior distribution by replacing the particle filter proposal within the IMH by Annealed Importance Sampling \cite{neal2001annealed}
or a sequential Monte Carlo sampler \cite{delmoral2006sequential}. This is illustrated in Appendix \ref{sec:smcsamplers}.

\newpage
\bibliographystyle{plain}
\bibliography{csmcrefresh}

\cleardoublepage
\appendix

\section{Appendix}

\subsection{Proof of Proposition \ref{prop:mtdist}\label{proof:prop:mtdist}}

We establish a monotonicity property of the proposed coupling for PIMH.
Monotonicity properties of IMH samplers were exploited in an exact simulation
context in \cite{corcoran2002perfect} without this explicit construction.
\begin{prop}
\label{prop:stochasticdominatingcoupling}Under the proposed coupling
scheme the sequence of likelihood estimates $(p_{N}^{(n+1)})_{n\ge0}$
stochastically dominates $(\tilde{p}_{N}^{(n)})_{n\ge0}$ in the sense
that for any $n\ge1$, $s\ge0$
\[
p_{N}^{(n)}\ge\tilde{p}_{N}^{(n-1)}\Rightarrow p_{N}^{(n+s)}\ge\tilde{p}_{N}^{(n+s-1)}\quad\text{a.s.}
\]
\end{prop}

\begin{proof}

The coupling procedure in Algorithm \ref{alg:couplepimh} uses a single proposal $p_{N}^{*}$ and samples $\mathfrak{u}\sim\mathcal{U}[0,1]$, with proposals being accepted according to:
\begin{align*}
\text{if}\;\mathfrak{u} & \le1\wedge\frac{p_{N}^{*}}{p_{N}^{(n)}}\;\text{then}\;p_{N}^{(n+1)}=p_{N}^{*}\text{, else}\;p_{N}^{(n+1)}=p_{N}^{(n)},\\
\text{if}\; \mathfrak{u} & \le1\wedge\frac{p_{N}^{*}}{\tilde{p}_{N}^{(n-1)}}\;\text{then}\;\tilde{p}_{N}^{(n)}=p_{N}^{*}\text{, else}\;\tilde{p}_{N}^{(n)}=\tilde{p}_{N}^{(n-1)}.
\end{align*}
We see that if $p_{N}^{(n)}\ge\tilde{p}_{N}^{(n-1)}$
then either
\begin{enumerate}
\item $p_{N}^{(n+1)}=p_{N}^{*}$ in which case
\[
\mathfrak{u}\le1\wedge\frac{p_{N}^{*}}{p_{N}^{(n)}}\implies\mathfrak{u}\le1\wedge\frac{p_{N}^{*}}{\tilde{p}_{N}^{(n-1)}}
\]
as $p_{N}^{(n)}\ge\tilde{p}_{N}^{(n-1)}$ so that $p_{N}^{(n+1)}=\tilde{p}_{N}^{(n)}=p_{N}^{*}$
(i.e. the chains meet).
\item $p_{N}^{(n+1)}=p_{N}^{(n)}$ and so $p_{N}^{*}\le p_{N}^{(n)}$. In this case, either $\tilde{p}_{N}^{(n)}=p_{N}^{*}$, and
so $p_{N}^{(n+1)}\ge\tilde{p}_{N}^{(n)}$,
or $\tilde{p}_{N}^{(n)}=\tilde{p}_{N}^{(n-1)}$ in which case both chains have rejected $p_{N}^{*}$ and
the ordering is preserved.
\end{enumerate}
Finally, from the initialization of the procedure,
we have $p_N^{(1)}\geq \tilde{p}_N^{(0)}$ because the initial
state of the second chain is used as a proposal in the first iteration
of the first chain.
\end{proof}

From the above reasoning we see that the chains meet when the first chain accepts
its proposal for the first time, as the second chain then necessarily accepts the same proposal.

From the initial state with likelihood estimate $p_N^{(0)}$,
the acceptance probability of the first chain is
$\int 1\wedge(p_{N}/p_{N}^{(0)}) \bar{g}(p_{N})\mathrm{d}p_{N}$,
with $\bar{g}$ denoting the density of the PF likelihood estimator $p_N$.
Thus, the time to the first acceptance
follows a Geometric distribution with success probability
$\int 1\wedge(p_{N}/p_{N}^{(0)}) \bar{g}(p_{N})\mathrm{d}p_{N}$.
The result stated in Proposition \ref{prop:mtdist} follows
when rewriting the problem using the error of the log-likelihood estimator
$\log\{p_{N}(y_{1:T})/p(y_{1:T})\}$.

\subsection{Integrated autocorrelation time for various test functions}
\label{sec:inefficiency}
We show here experimentally that $\text{IF}(h)$ is approximately proportional to $\text{IF}(\sigma)$ for various test functions: $h_1:x_{1:T}\mapsto x_1$,  $h_2:x_{1:T}\mapsto x_T$,  $h_3:x_{1:T}\mapsto \sum_t x_t$ and  $h_4:x_{1:T}\mapsto \sum_t x_t^2$. This is illustrated in Figure \ref{fig:toyineff} where $\text{IF}(h)$ is displayed for a range of $N$  against $\text{IF}(\sigma)$  over the corresponding range of $\sigma$.

\begin{figure}[ht]
\begin{centering}
\includegraphics[scale=0.4]{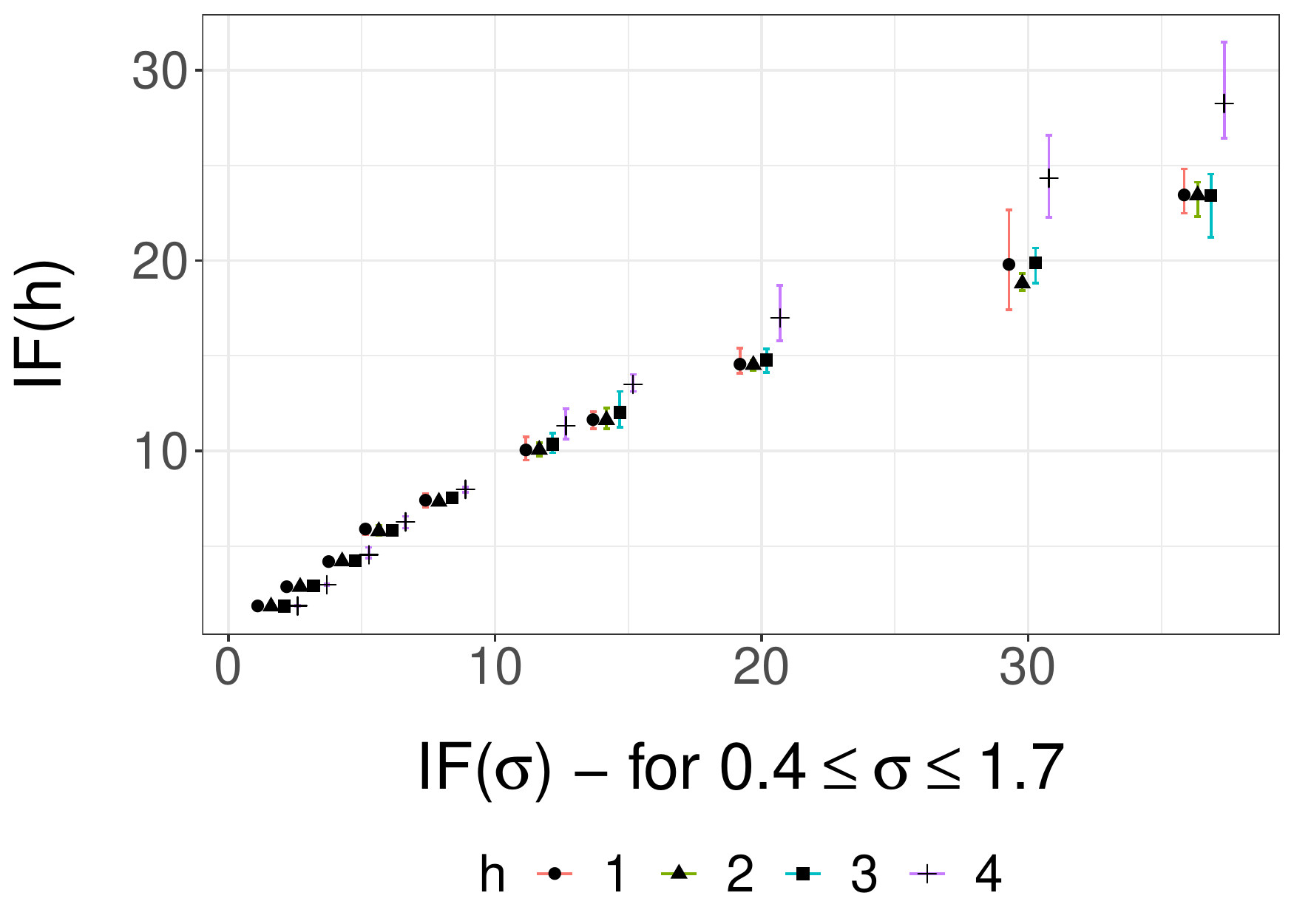}
\caption{Inefficiency ${\text{IF}[h_i]}$ versus ${\text{IF}[\sigma]}$, with markers indicating the test functions $h_1,h_2,h_3,h_4$. The vertical axis scale is relative, depending on the test function.}
\label{fig:toyineff}
\end{centering}
\end{figure}

\begin{figure}
\begin{centering}
\includegraphics[scale=0.4]{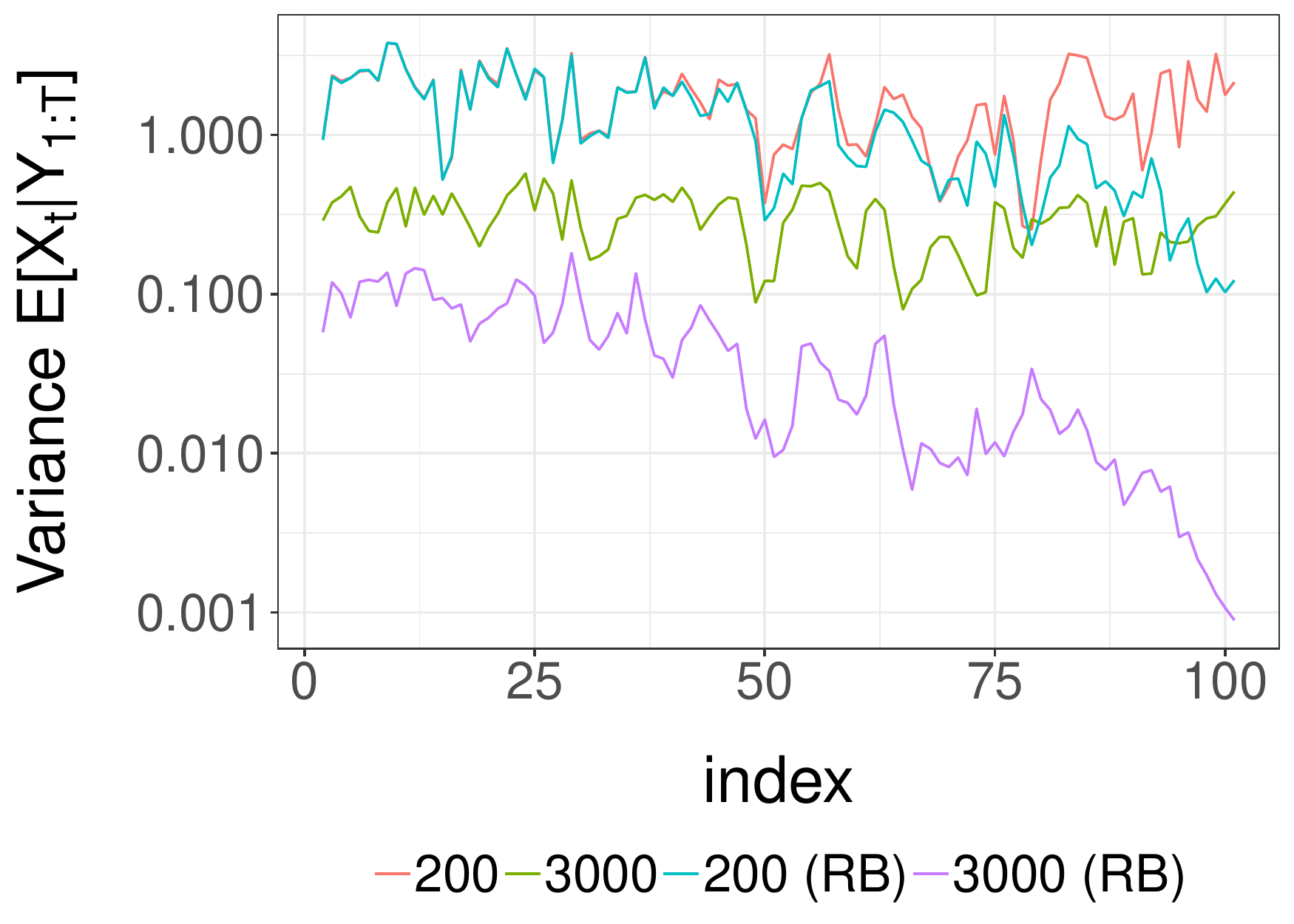}
\end{centering}
\caption{Empirical variance of unbiased estimators of $\mathbb{E}(X_{1,\Delta t}\vert y_{1:T})$: $H_{0:0}$ and Rao-Blackwellised (RB) estimator $\bar{H}_{0:0}$ for stochastic kinetic model.\label{fig:rbcomp}}
\end{figure}

\subsection{Rao-Blackwellisation for stochastic kinetic model}
\label{subsec:raoblackstochkinetic}
We demonstrate here the gains arising from the use of a Rao-Blackwellized estimator detailed in Section \ref{sec:coupledPIMH}.
We display in Figure \ref{fig:rbcomp} the variance of the two unbiased estimators of $\mathbb{E}(X_{1,\Delta t}\vert y_{1:T})$ for $t=\Delta,...,T\Delta$ and $T=100$ for the latent Markov jump process and two different values of $N$, $200$ and $3,000$.
In both cases we set $k=m=0$. As expected $\bar{H}_{0:0}$ outperforms $H_{0:0}$ but the benefits are much higher for $t$ close to $T$ than when $t$ is close to $1$.
For example, we see that for $N=200$ the estimators $H_{0:0}$ and $\bar{H}_{0:0}$ coincide for $t\leq 35$.
This is an expected consequence of the particle path degeneracy problem \cite{doucet2009tutorial,jacob2015path,kantas2015particle}, with many particles $(X_{1:T}^{i})_{i\in[N]}$ obtained by the PF at time $T$ sharing common ancestors for $t$ close to 1 when $N$ is too small; see \cite{jacob2015path} for results on the corresponding coalescent time.

\begin{figure}
\begin{centering}
\subfloat[Daily returns of S\&P 500 data]{\begin{centering}
\includegraphics[scale=0.4]{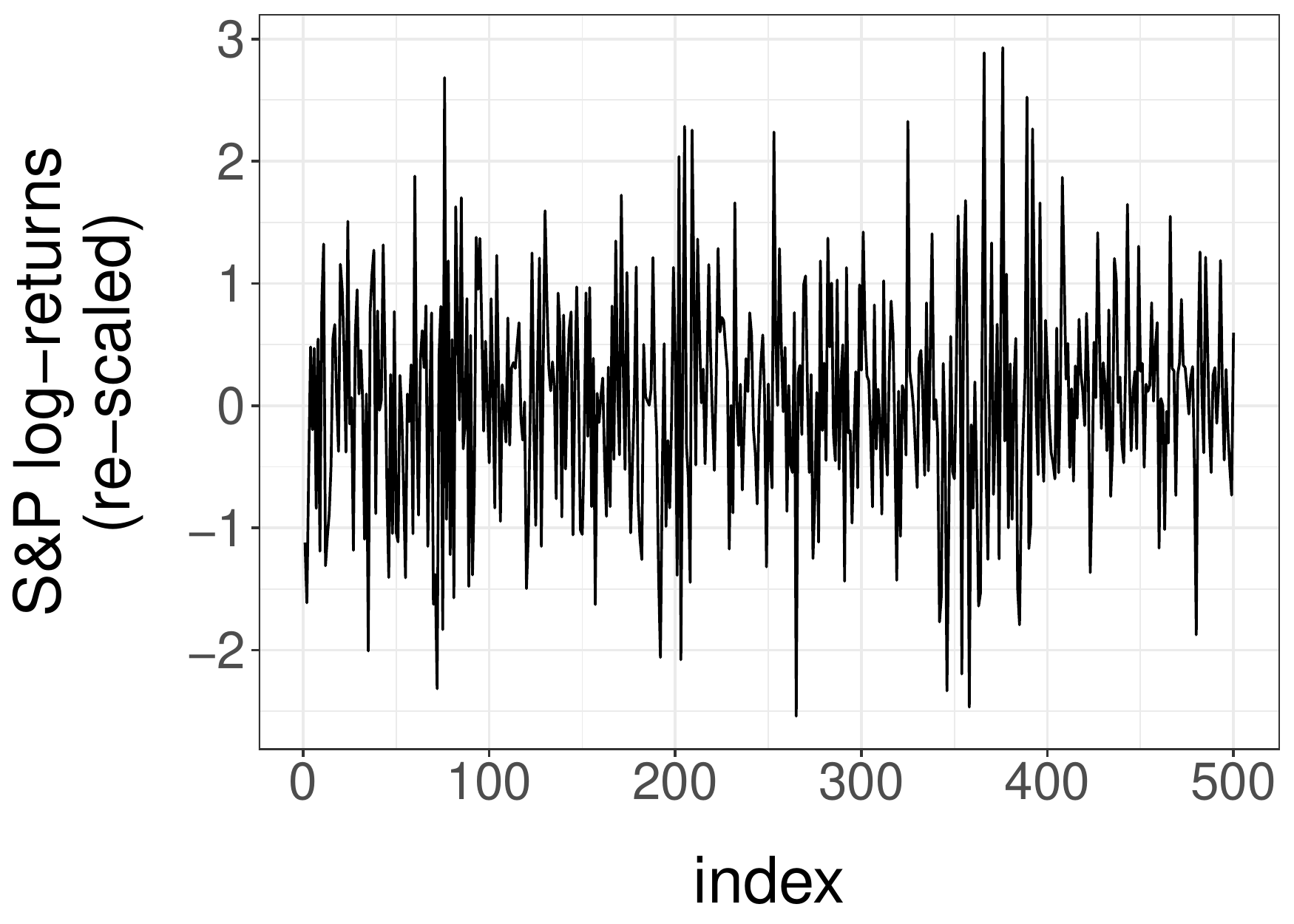}
\par\end{centering}
\centering{}\label{fig:data-1}}

\subfloat[Particle estimates of the log-likelihood function]{\begin{centering}
\includegraphics[scale=0.4]{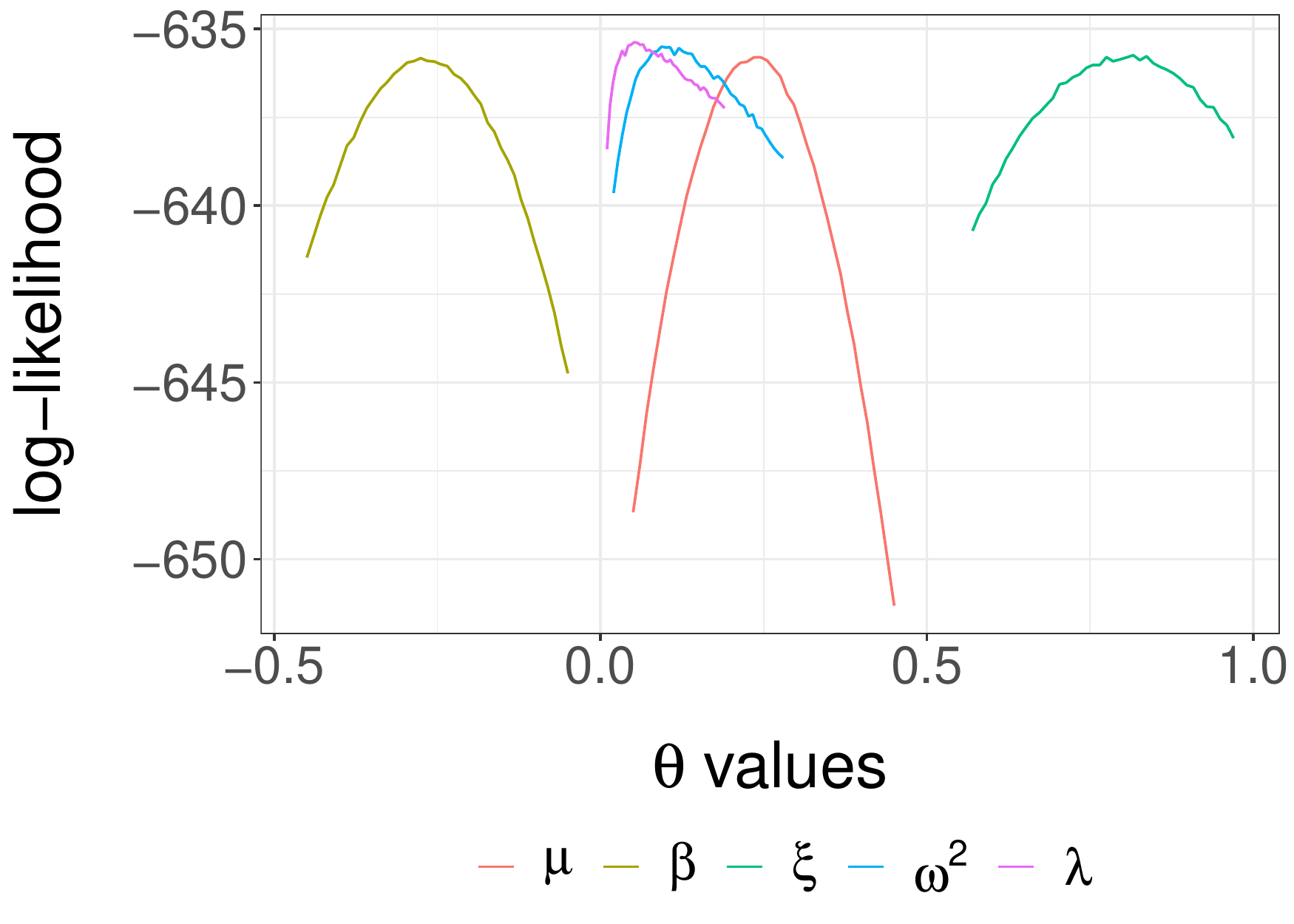}
\par\end{centering}
\centering{}\label{fig:ll-1}}
\par\end{centering}
\centering{}\caption{Data and parameter estimation for L\'evy-driven stochastic volatility
model.}
\end{figure}

\subsection{Data and parameter estimation}
\label{subsec:datasummary}
The data used comprised of $T=500$ log-returns of the S\&P 500, starting
from the $3^{rd}$ January 2005, with the scaled data taken from the
stochastic volatility example used to demonstrate $\text{SMC}^{2}$
in \cite{chopin2013smc2}. We plot the raw data in Figure \ref{fig:data-1}.
Parameters were estimated using a two-stage procedure, with $\text{SMC}^{2}$
used to find a region of high marginal likelihood under the model.
Further refinement was performed to compute the maximum likelihood estimator (MLE)
$\hat{\theta}$ of $\theta=(\mu,\beta,\xi,\omega^{2},\lambda)$ using a grid search around values close to the optimum
using $N=10,000$ particles. We obtained $\hat{\theta}=(0.24,-0.28,0.82,0.09,0.05)$.
Likelihood curves around the optimal values are shown in Figure \ref{fig:ll-1},
where for each parameter component the log-likelihood was varied while
keeping the other parameters fixed at $\hat{\theta}$.

\section{Application of unbiased estimation to SMC samplers}
\label{sec:smcsamplers}
SMC samplers are a class of SMC algorithms that can be used in Bayesian inference
to approximate expectations w.r.t. complex posteriors for static models \cite{delmoral2006sequential}.
We show here how we can directly use the methodology proposed in this paper to obtain unbiased estimators
of these expectations.

\subsection{Bayesian computation using SMC samplers}
Assume one is interested in sampling from the posterior density $\pi(x)\propto\nu(x) L(x)$ where
$\nu(x)$  the prior density w.r.t. a suitable dominating measure and $L(x)$ is the likelihood. We also assume that one can
sample from  $\nu$. To approximate $\pi$, a specific version of SMC samplers introduces a sequence of $T-1$ intermediate
densities $\pi_t$ for $t=2,...,T$ bridging $\nu$ to $\pi$ using
\[
\gamma_{t}(x)=\nu(x) L^{\beta_t}(x),\qquad\pi_{t}(x)=\frac{\gamma_{t}(x)}{\mathcal{Z}_t},
\]
where $\beta_1=0<\beta_2<...<\beta_{T}=1$. The choice of the sequence $\{\beta_t:t=2,...,T-1\}$ can be guided using a preliminary adaptive SMC scheme, as in \cite{zhou2016toward}, which should subsequently be fixed to preserve unbiasedness of the normalizing constant estimate and validity of the resulting PIMH.
In SMC samplers, particles are initialized at time $t=1$ by sampling from the prior ensuring $w_1(x_1)=1$. At time $t\geq2$, particles are sampled according to an MCMC
kernel leaving $\pi_{t-1}$ invariant and are then weighted according to
\[
w_t(x_{t-1},x_{t})=\frac{\gamma_{t}(x_{t-1})}{\gamma_{t-1}(x_{t-1})}=L^{\beta_{t}-\beta_{t-1}}(x_{t-1}).
\]
Particles are resampled according to these weights and we set
\[
\mathcal{Z}_{t,N}=\mathcal{Z}_{t-1,N}\cdot\frac{1}{N}\sum_{i=1}^{N}w_{t}(X_{t-1}^{A_{t-1}^{i}},X_{t}^{i}),
\]
with $\mathcal{Z}_{1,N}=1$.
At time $T$, $\pi_{N}(\mathrm{d}x_{T}):=\sum_{i=1}^{N}W_{T}^{i}\delta_{X_{T}^{i}}(\mathrm{d}x_{T})$ provides a Monte Carlo approximation of the distribution $\pi$ and $\mathcal{Z}_{T,N}$ plays the role of $p_N(y_{1:T})$, approximating the normalizing constant
$\mathcal{Z}_T$ of $\pi=\pi_T$. If no resampling is used, this specific version of SMC samplers coincides with AIS \cite{neal2001annealed} in which case $\mathcal{Z}_{T,N}$ is given by the average of the product of the incremental weights from time $t=1$ to $t=T$ instead of the product of the averaged incremental weights. We can use this SMC sampler algorithm or AIS directly within the coupled PIMH scheme, replacing $p_N(y_{1:T})$ by $\mathcal{Z}_{T,N}$ in the acceptance probabilities.  We see that $\sup_{(x,x')\in\mathsf{X}^{2}}w_t(x,x')<\infty$
provided that $\sup_{x\in\mathsf{X}}L(x)<\infty$. Under this condition, if Assumption \ref{assumption:marginaldistributionsjoint}
is satisfied then the estimator $\bar{H}_{k:m}$ of $\pi(h)$ is unbiased and has finite variance and finite expected
cost.

\subsection{Numerical example}
We use here coupled PIMH to debias expectations w.r.t. the posterior
distribution for a Bayesian mixture model discussed in \cite{lee2010utility}.
We have
\[
L(x)=\prod_{n=1}^{M}\left(\frac{1}{D}\sum_{i=1}^{D}\mathcal{N}(y_{n};x^{i},\sigma^{2})\right)
\]
with $x:=(x^{1},...,x^{D})\in\mathbb{R}^{D}$ constituting the unknown mean components. We consider here $D=4$ mixture components and $M=100$ observations.
A uniform prior distribution is placed on $x$ over the hypercube $[-10,10]^{D}$. The resulting posterior distribution is multimodal.
We set $\sigma=1$ and simulate observations from the model with true values $x^{*}=(-3,0,3,6)$.
We adopt a symmetric random walk for the Metropolis--Hastings proposals
with identity covariance and pick $\beta_{t}=\left(\frac{t-1}{T-1}\right)^{2}$
for $T=200$.

We simulate $10,000$ estimators with $m=64$ and $N=100$, after
which we are able to estimate variance of test functions for a range
of values of $k$ and $m$ noting that there will be some correlation
introduced between estimators. The results are shown in Figure \ref{fig:smcsampler}
where we plot the meeting times of the unbiased estimators in Figure
\ref{fig:mtsmc} and the variance of the estimators for a range of
values of $m$ in Figure \ref{fig:varestimatorssmc} using $h:x\mapsto x^{1}+x^{2}+(x^{1})^{2}+(x^{2})^{2}$.
Uncertainty in the estimated values of $\mathbb{V}[\bar{H}_{k:m}(h)]$ was obtained using $1,000$ bootstrap samples, resampling $10,000$ of the $10,000$ unbiased estimators with replacement.
The figure shows the variance of these estimators for a range of values
of $m\in\{4,8,...,64\}$ while varying $k\in\{0,...,m-1\}$. We see,
firstly, as expected that as $m$ increases the variance of the estimators
reduces. Secondly, for each value of $m$ we see that there exists
an optimal value of $k$, however, as $m$ increases the optimum becomes
less pronounced, suggesting that as $m$ increases there is a degree
of insensitivity to the choice of $k$. Finally, for the range of
$m$ considered, the optimal values of $k$ appear in a comparatively
small interval close to the origin, suggesting that
it is not necessary to use large values of $k$ to reduce the variance contribution arising from the bias correction.
\begin{figure}
\subfloat[Empirical distribution of meeting times for SMC sampler with $N=100$ over 10,000 independent runs. The estimated $95^{th}$
and $99^{th}$ percentiles were 6 and 13 respectively.]{\begin{centering}
\includegraphics[scale=0.4]{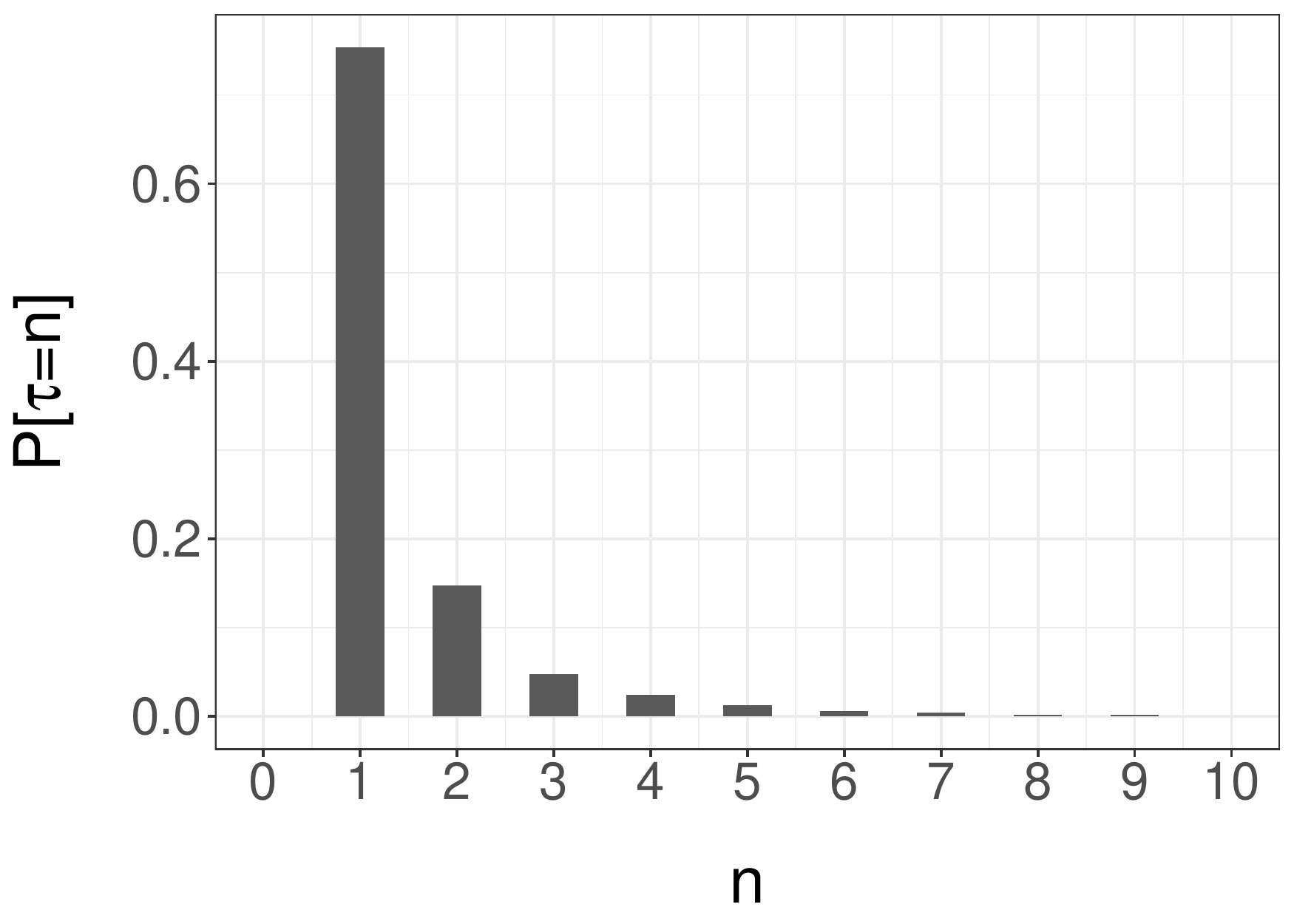}
\par\end{centering}
\label{fig:mtsmc}}

\subfloat[{$\mathbb{V}[\bar{H}_{k:m}]$ of Rao-Blackwellised unbiased estimators of $\pi(h)$
for SMC sampler as a function of $k$ for a range of values of $m$.
The shaded regions correspond to the $1^{st}$ and $99^{th}$ percentiles of the variance estimator.}]{\begin{centering}
\includegraphics[scale=0.4]{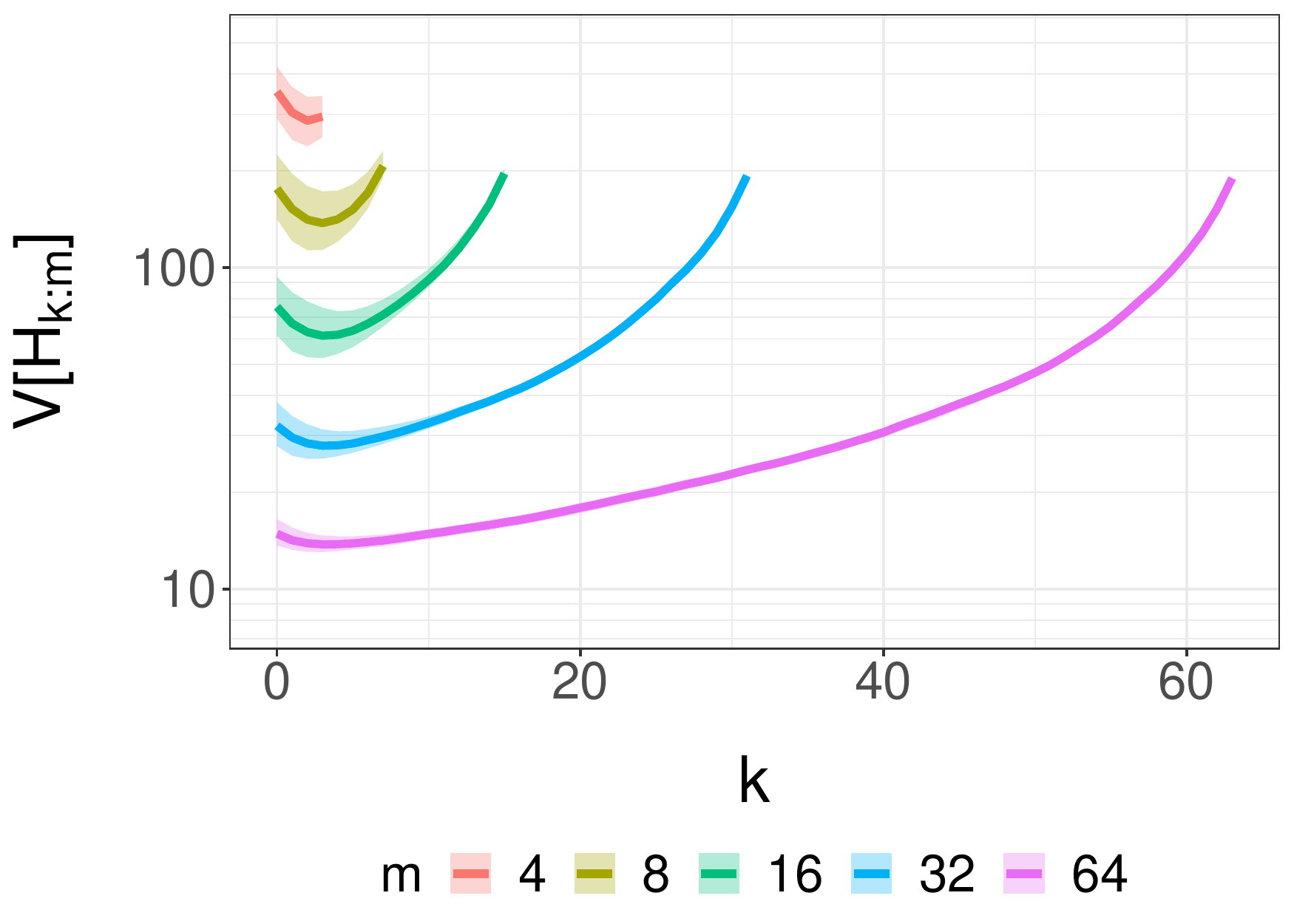}
\end{centering}
\label{fig:varestimatorssmc}}
\centering{}\caption{SMC sampler unbiased estimators}
\label{fig:smcsampler}
\end{figure}

\end{document}